\documentclass[10pt,conference]{IEEEtran}

\usepackage{booktabs}
\usepackage{amsmath}
\usepackage{amsthm}
\usepackage{amsfonts}
\usepackage[utf8]{inputenc}

\usepackage{times}
\usepackage{graphicx}
\usepackage{url}
\usepackage{epsfig}
\usepackage{algorithm}
\usepackage{algorithmicx}
\usepackage[noend]{algpseudocode}
\usepackage{xspace}
\usepackage{multirow}
\usepackage{subcaption}
\usepackage{thmtools, thm-restate}
\usepackage{xcolor}

\usepackage[sort,compress]{cite}
\usepackage[
breaklinks=true,
unicode=true,
urlcolor = blue,
colorlinks = true,
citecolor = blue,
linkcolor = blue]{hyperref}

\newcommand{\ignore}[1]{}

\newif\ifsubmit
\submittrue

\newif\ifgrad
\gradfalse

\ifsubmit
\newcommand{\chang}[1]{}
\else
\newcommand{\chang}[1]{{\authnote{\textcolor{red}{Chang: #1}}}}
\fi

\title{Manipulating Machine Learning: Poisoning Attacks and Countermeasures for Regression Learning}

\author{
\IEEEauthorblockN{Matthew Jagielski\IEEEauthorrefmark{1},
                  Alina Oprea\IEEEauthorrefmark{1},
                  Battista Biggio \IEEEauthorrefmark{2}\IEEEauthorrefmark{3},
                  Chang Liu\IEEEauthorrefmark{4},
                  Cristina Nita-Rotaru\IEEEauthorrefmark{1},
                  and Bo Li\IEEEauthorrefmark{4}}

\and\and
\IEEEauthorblockA{
    \IEEEauthorrefmark{1}Northeastern University, Boston, MA}
\and
\IEEEauthorblockA{\IEEEauthorrefmark{2}University of Cagliari, Italy}
\and \IEEEauthorblockA{\IEEEauthorrefmark{3}Pluribus One, Italy}

\and
\IEEEauthorblockA{
    \IEEEauthorrefmark{4}UC Berkeley, Berkeley, CA}
}
\newcommand{\diff}[2]{\frac{\partial #1}{\partial #2}}
\newcommand{\vct}[1]{\ensuremath{\boldsymbol{#1}}} 
\newcommand{\mat}[1]{\ensuremath{\mathbf{#1}}}
\newcommand{\set}[1]{\ensuremath{\mathcal{#1}}}

\newcommand{\T}{\ensuremath{\top}}

\newcommand{\myparagraph}[1]{\smallskip \noindent \textbf{#1.}}
\newcommand{\ie}{{i.e.}\xspace}
\newcommand{\eg}{{e.g.}\xspace}
\newcommand{\etal}{{et al.}\xspace}

\newcommand\trainset{\ensuremath{\set D_{\rm tr}}}
\newcommand\valset{\ensuremath{\set D_{\rm val}}}
\newcommand\trainsetsub{\ensuremath{\set D^\prime_{\rm tr}}}
\newcommand\poisonset{\ensuremath{\set D_p}}
\newcommand\alltrain{\ensuremath{\set D}}
\newcommand\attackloss{\ensuremath{\mathcal{W}}}
\newcommand\defenderloss{\ensuremath{\mathcal{L}}}
\newcommand\thetapar{\ensuremath{\vct \theta}}
\newcommand\thetaopt{\ensuremath{\vct \theta^\star}}
\newcommand\thetaoptpoison{\ensuremath{\vct \theta_p^\star}}
\newcommand\npois{\ensuremath{p}}

\newcommand\wb{\ensuremath{{\vct w}}}
\newcommand\xb{\ensuremath{{\vct x}}}
\newcommand\Xb{\ensuremath{{\mat X}}}
\newcommand\attackpar{\ensuremath{{\vct z}}}

\newcommand\baseline{\ensuremath{\mathsf{BGD}}}

\newcommand\rgdrf{\ensuremath{\mathsf{rGDRF}}}
\newcommand\statp{\ensuremath{\mathsf{StatP}}}
\newcommand\rmml{\ensuremath{\mathsf{StatP}}}

\newcommand\optp{\ensuremath{\mathsf{OptP}}}

\newcommand\trim{\ensuremath{\mathsf{TRIM}}}
\newcommand\RANSAC{\ensuremath{\mathsf{RANSAC}}}
\DeclareMathOperator{\argmin}{arg\,min}
\DeclareMathOperator{\argmax}{arg\,max}

\newcommand\invflip{\ensuremath{\mathsf{InvFlip}}}
\newcommand\bflip{\ensuremath{\mathsf{BFlip}}}

\newcommand\trainobj{\ensuremath{\attackloss_{\rm tr}}}
\newcommand\valobj{\ensuremath{\attackloss_{\rm val}}}

\newcommand{\revision}[1]{{\textcolor{black} {#1}}}

\begin{document}

\maketitle

\begin{abstract}
As machine learning becomes widely used for automated decisions, attackers have strong incentives to manipulate the results and models generated by machine learning algorithms. In this paper, we perform the first systematic study of poisoning  attacks and their countermeasures for linear regression models. In poisoning attacks, attackers deliberately influence the training data to manipulate the results of a predictive model. We propose a theoretically-grounded optimization framework specifically designed for linear regression and demonstrate its effectiveness on a range of datasets and models. We also introduce a fast statistical attack that requires limited knowledge of the training process. Finally, we design a new principled defense method that is highly resilient against all poisoning attacks. We provide formal guarantees about its convergence and an upper bound on the effect of poisoning attacks when the defense is deployed. We evaluate extensively our attacks and defenses on three realistic datasets from health care, loan assessment, and real estate domains.\footnote{Preprint of the work accepted for publication at the 39th IEEE Symposium on Security and Privacy, San Francisco, CA, USA, May 21-23, 2018.} \footnote{Note: a prior version of this paper had a bug in the implementation of the \trim\ algorithm, leading to inaccurately good results. After the bugfix, \trim\ still significantly outperforms prior defenses.}
\end{abstract} 

\section{Introduction}

As more applications with large societal impact rely on machine learning for automated decisions, several concerns have emerged about potential vulnerabilities introduced by machine learning algorithms. Sophisticated attackers have strong incentives to manipulate the results and models generated by machine learning algorithms to achieve their objectives. For instance, attackers can deliberately influence the training dataset to manipulate the results of a predictive model (in \emph{poisoning} attacks~\cite{Perdisci06,newsome2006paragraph,Nelson08,ANTIDOTE,Biggio2012poisoning,Newell14,Xiao15}), cause mis-classification of new data in the testing phase (in \emph{evasion attacks}~\cite{biggio13-ecml,Szegedy14,Goodfellow14,Srndic14,Papernot16,Papernot17,Carlini17}) or infer private information on training data (in \emph{privacy attacks}~\cite{fredrikson2014privacy,Membership,Fredrikson15}). Several experts from academia and industry highlighted the importance of considering these vulnerabilities in designing machine learning systems in a recent hearing held by the Senate Subcommittee on Space, Science, and Competitiveness entitled ``The Dawn of AI''~\cite{DawnAI}. The field of \emph{adversarial machine learning} studies the effect of such  attacks against machine learning models and aims to design robust defense algorithms~\cite{Huang2011adversarial}.
A comprehensive survey can be found in~\cite{wildpatterns}.

We consider the setting of poisoning attacks here, in which attackers inject a small number of corrupted points in the training process. Such poisoning attacks have been practically demonstrated in worm signature generation~\cite{newsome2006paragraph,Perdisci06}, spam filters~\cite{Nelson08}, DoS attack detection~\cite{ANTIDOTE}, PDF malware classification~\cite{Xiao15}, handwritten digit recognition~\cite{Biggio2012poisoning}, and sentiment analysis~\cite{Newell14}. We argue that these attacks become easier to mount today as many machine learning models need to be updated regularly to account for continuously-generated data. Such scenarios require \emph{online training}, in which machine learning models are updated based on new incoming training data. For instance, in cyber-security analytics, new Indicators of Compromise (IoC) rise due to the natural evolution of malicious threats, resulting in updates to machine learning models for threat detection~\cite{Predator}. These IoCs are collected from online platforms like VirusTotal, in which attackers can also submit IoCs of their choice. In personalized medicine, it is envisioned that patient treatment is adjusted in real-time by analyzing information crowdsourced from multiple participants~\cite{MLHealth}.  By controlling a few devices, attackers can submit fake information (e.g., sensor measurements), which is then used for training models applied to a large set of patients. Defending against such poisoning attacks is challenging with current techniques. Methods from robust statistics (e.g, \cite{Huber64,RANSAC}) are resilient against noise but perform poorly on adversarially-poisoned data, and methods for sanitization of training data operate under restrictive adversarial models~\cite{Ciocarlie08}.

One fundamental class of supervised learning is linear regression. Regression is widely used for prediction in many settings (e.g., insurance or loan risk estimation, personalized medicine, market analysis).  In a regression task a numerical \emph{response variable} is predicted using a number of \emph{predictor variables}, by learning a model that minimizes a \emph{loss function}. Regression is powerful as it can also be used for classification tasks by mapping numerical predicted values into class labels. Assessing the real impact of adversarial manipulation of training data in linear regression, as well as determining how to design learning algorithms resilient under strong adversarial models is not yet well understood.

In this paper, we conduct the first systematic study of poisoning attacks and their countermeasures for linear regression models.
We make the following contributions: (1) \revision{we are the first to consider the problem of poisoning linear regression  under different adversarial models;} (2) \revision{starting from an existing baseline poisoning attack for classification, we propose a theoretically-grounded optimization framework specifically tuned for regression models;} (3) \revision{we design a fast statistical attack that requires minimal knowledge on the learning process;}  (4) we propose a principled defense algorithm with significantly increased robustness than known methods against a large class of attacks; (5) we extensively evaluate our attacks and defenses on four regression models (OLS, LASSO, ridge, and elastic net), and on several datasets from different domains, including health care, loan assessment, and real estate. We elaborate our contributions below.

\vspace{0.05cm}
\noindent $\bullet$  \revision{On the attack dimension, we are the first to consider the problem of poisoning attacks against linear regression models. Compared to classification poisoning, in linear regression the response variables are continuous and their values also can be selected by the attacker. First, we adapt an existing poisoning attack for classification~\cite{Xiao15} into a baseline regression attack. Second, we design an optimization framework for regression poisoning in which the initialization strategy, the objective function, and the optimization variables can be selected to maximize the attack's impact on a particular model and dataset.  Third, we introduce a fast statistical attack that is motivated by our theoretical analysis and insights. }We find that optimization-based attacks are in general more effective than statistical-based techniques, at the expense of higher computational overhead and more information required by the adversary on the training process.

\vspace{0.05cm}
\noindent $\bullet$ On the defense axis, we propose a principled approach to constructing a defense algorithm called \trim, which provides high robustness and resilience against a large class of poisoning attacks. The \trim\ method estimates the regression parameters iteratively, while using a trimmed loss function to remove points with large residuals. After few iterations, \trim\ is able to isolate most of the poisoning points  and learn a robust regression model. \trim\ performs significantly better and is much more effective in providing robustness compared to known methods from robust statistics (Huber~\cite{Huber64} and RANSAC~\cite{RANSAC}), typically designed to provide resilience against noise and outliers. In contrast to these methods, \trim\ is resilient to poisoned points with similar distribution as the training set. \trim\ also outperforms other robust regression algorithms designed for adversarial settings (e.g., Chen et al.~\cite{Chen13} \revision{and RONI~\cite{Nelson08}}). We provide theoretical guarantees on the convergence of the algorithm and an upper bound on the model Mean Squared Error (MSE) generated when a fixed percentage of poisoned data is included in the training set.

\vspace{0.05cm}
\noindent $\bullet$ We evaluate our novel attacks and defenses extensively on four linear regression models and three datasets from health care, loan assessment, and real estate domains. First, we demonstrate the significant improvement of our attacks over the baseline attack of Xiao et al. in poisoning all models and datasets. For instance, the MSEs of our attacks are increased \revision{{\bf by a factor of 6.83}} compared to the Xiao et al. attack, and \revision{{\bf a factor of 155.7}} compared to unpoisoned regression models.  In a case study health application, we find that our attacks can cause devastating consequences. The optimization attack causes 75\% of patients’ Warfarin medicine dosages to change by an average of \revision{93.49\%}, while one tenth of these patients have their dosages changed by \revision{358.89\%}. Second, we show that our defense \trim\ is also significantly more robust than existing methods against all the attacks we developed. The median MSE increase over all attacks, models, and datasets is \textbf{only 6.1\%}, and only 20\% of attacks cause more than 27.2\% MSE increase. \revision{\trim\ achieves MSEs much lower than existing methods, improving Huber by a factor of 131.8, RANSAC by a factor of 17.5, and RONI by a factor of 20.28, on one dataset.}

\vspace{0.05cm}
\noindent {\bf Outline.} We start by providing background on regression learning, as well as introducing our system and adversarial model in Section~\ref{sec:model}. We describe the baseline attack adapted from  Xiao et al.~\cite{Xiao15}, and our new poisoning attacks in Section~\ref{sec:attacks}. Subsequently, we introduce our novel defense algorithm \trim\ in Section~\ref{sec:defense}. Section~\ref{sec:eval} includes a detailed experimental analysis of our attacks and defenses, as well as comparison with previous methods. Finally, we present related work in Section~\ref{sec:related} and conclude in Section~\ref{sec:conclusions}.

\section{System and adversarial model}
\label{sec:model}

Linear regression is at the basis of machine learning~\cite{ESL}. It is widely studied and applied in many applications due to its efficiency, simplicity of use, and effectiveness.  Other more advanced learning methods (e.g., logistic regression, SVM, neural networks) can be seen as generalizations or extensions of linear regression. We systematically study the effect of poisoning attacks and their defenses for linear regression. We believe that our understanding of the resilience of this fundamental class of learning model to adversaries will enable future research on other classes of supervised learning methods.

\myparagraph{Problem definition} Our system model is a supervised setting consisting of a \emph{training phase} and a \emph{testing phase} as shown in Figure~\ref{fig:model} on the left (``Ideal world''). The learning process includes a data pre-processing stage that performs data cleaning and normalization, after which the training data can be represented, \revision{without loss of generality}, as $\trainset = \{( \vct x_i, y_i)\}_{i=1}^{n}$, where $\vct x_i \in [0,1]^{d}$ are $d$-dimensional numerical \emph{predictor variables} (or \emph{feature vectors}) and $y_i \in [0,1]$ are numerical \emph{response variables}, for $i \in \{1,\dots, n \}$. After that, the learning algorithm is applied to generate the \emph{regression model} at the end of the training phase. In the testing phase, the model is applied to new data after pre-processing, and a numerical predicted value is generated using the regression model learned in training. Our model thus captures a standard multi-dimensional regression setting \revision{applicable to} different prediction tasks. 

In linear regression, the model output at the end of the training stage is a linear function
$f(\xb, \thetapar)=\wb^\T \xb +b$, which predicts the value of $y$ at $\vct x$.
This function is parametrized by a vector $\thetapar = (\vct w, b) \in \mathbb R^{d+1}$ consisting of
the feature weights $\wb \in \mathbb R^d$ and the bias $b \in \mathbb R$.
\revision{
Note that regression is substantially different from classification, as the $y$ values are numerical, rather than being a set of indices (each denoting a different class from a predetermined set).}
The parameters of $f$ are chosen to minimize a quadratic loss function:
\begin{equation}
\defenderloss(\trainset, \thetapar)= \underbrace{\textstyle\frac{1}{n}\textstyle \sum_{i=1}^{n} \left( f(\xb_i, \thetapar)-y_i \right)^2}_{\rm MSE(\set \trainset, \thetapar)} + \lambda\Omega(\wb) \, ,
\label{eq:learning}
\end{equation}
\noindent where the Mean Squared Error ${\rm MSE}(\trainset, \thetapar)$ measures the error in the predicted values assigned by $f(\cdot, \thetapar)$ to the training samples in $\trainset$ as the sum of squared residuals, $\Omega(\wb)$ is a regularization term \revision{penalizing large weight values}, and $\lambda$ is the so-called regularization parameter.
\revision{
Regularization is used to prevent \emph{overfitting}, \ie, to preserve the ability of the learning algorithm to \emph{generalize} well on unseen (testing) data. For regression problems, this capability, \ie, the expected performance of the trained function $f$ on unseen data, is typically assessed by measuring the MSE on a separate test set.
}
Popular linear regression methods differ mainly in the choice of the regularization term. In particular, we consider four models in this paper:
\begin{enumerate}
\item {\bf Ordinary Least Squares (OLS)}, for which $\Omega(\wb)  = 0$ (\ie, no regularization is used);
\item {\bf Ridge regression}, which uses $\ell_2$-norm regularization \revision{$\Omega(\wb)  = \frac{1}{2}\|\wb\|_2^2$};
\item {\bf LASSO}, which uses $\ell_1$-norm regularization $\Omega(\wb)  = \|\wb\|_1$;
\item {\bf Elastic-net regression}, which uses a combination of $\ell_1$-norm and $\ell_2$-norm regularization \revision{$\Omega(\wb)  = \rho \|\wb\|_1 + (1-\rho) \frac{1}{2}\|\wb\|_2^2$}, where $\rho \in (0,1)$ is a configurable parameter, commonly set to 0.5 (as we do in this work).
\end{enumerate}

When designing a poisoning attack, we consider two  metrics for quantifying the effectiveness of the attack. First, we measure the success rate of the poisoning attack by the difference in testing set MSE of the corrupted model compared to the legitimate model (trained without poisoning). Second, we consider the running time of the attack.

\begin{figure}[t]
\begin{centering}
\includegraphics[width=7cm]{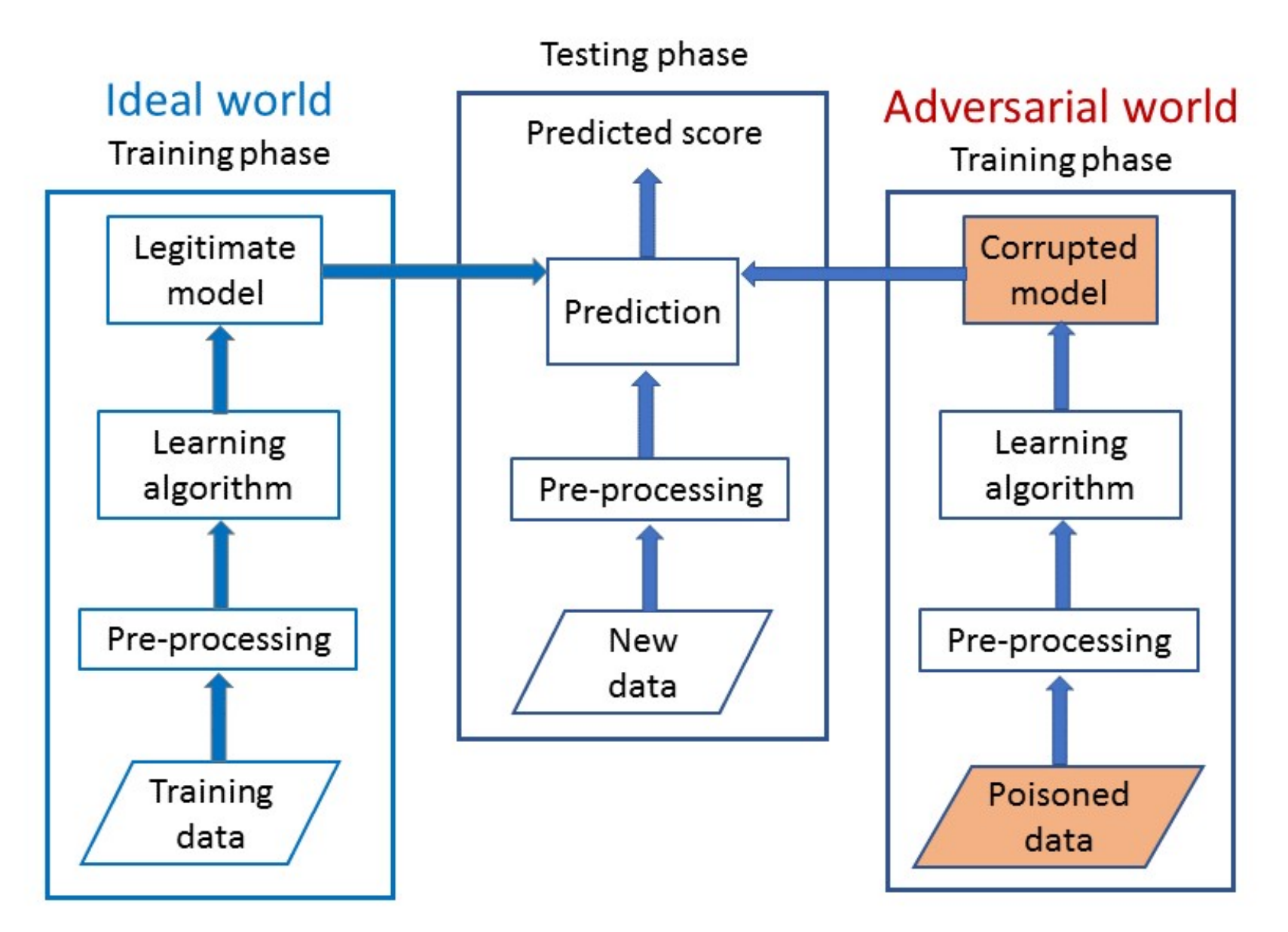}
\caption{System architecture.}
\label{fig:model}
\end{centering}
\end{figure}

\subsection{Adversarial model}

\revision{
We provide here a detailed adversarial model for poisoning attacks against regression algorithms, inspired from previous work in~\cite{Xiao15,biggio17-aisec,biggio14-tkde,Huang2011adversarial}.
The model consists of defining the adversary's goal, knowledge of the attacked system, and capability of manipulating the training data, to eventually define an optimal poisoning attack strategy.}

\myparagraph{Adversary's Goal} \revision{The goal of the attacker is to corrupt the learning model generated in the training phase, so that predictions on new data will be modified in the testing phase.
The attack is considered a \emph{poisoning availability attack}, if its goal is to affect prediction results indiscriminately, \ie, to cause a denial of service.
It is instead referred to as a \emph{poisoning integrity attack}, if the goal is to cause specific mis-predictions at test time, while preserving the predictions on the other test samples.
This is a similar setting to that of backdoor poisoning attacks recently reported in classification settings~\cite{gu17,chen17}.} 

\myparagraph{Adversary's Knowledge} We assume here two distinct attack scenarios, referred to as \emph{white-box} and \emph{black-box} attacks in the following.
In \emph{white-box attacks}, the attacker is assumed to know the training data $\trainset$, the feature values $\vct x$, the learning algorithm $\defenderloss$, and even the trained parameters $\vct \theta$.
These attacks have been widely considered in previous work, although mainly against classification algorithms~\cite{Biggio2012poisoning,Xiao15,mei15-aaai}.
In \emph{black-box attacks}, the attacker has no knowledge of the training set $\trainset$ but can collect a substitute data set $\trainsetsub$. The feature set and learning algorithm are known, while the trained parameters are not. However, the latter can be estimated by optimizing $\defenderloss$ on the substitute data set $\trainsetsub$. This setting is useful to evaluate the \emph{transferability} of poisoning attacks across different \emph{training sets}, as discussed in~\cite{biggio17-aisec,Xiao15}.

\myparagraph{Adversary's Capability} In poisoning attacks, the attacker injects poisoning points into the training set before the regression model is trained (see the right side of Figure~\ref{fig:model} labeled ``Adversarial world''). The attacker's capability is normally limited by upper bounding the number $\npois$ of \emph{poisoning points} that can be injected into the training data, whose feature values and response variables are arbitrarily set by the attacker within a specified range (typically, the range covered by the training data, \ie, $[0,1]$ in our case)~\cite{Xiao15,biggio17-aisec}.
The total number of points in the poisoned training set is thus $N=n+\npois$, with $n$ being the number of pristine training samples. We then define the ratio $\alpha = \npois/n$, and the \emph{poisoning rate} as the actual fraction of the training set controlled by the attacker, \ie, $n/N = \alpha/(1+\alpha)$.
\revision{In previous work, poisoning rates higher than 20\% have been only rarely considered, as the attacker is typically assumed to be able to control only a \emph{small} fraction of the training data. This is motivated by application scenarios such as crowdsourcing and network traffic analysis, in which attackers can only reasonably control a \emph{small} fraction of participants and network packets, respectively.
Moreover, learning a sufficiently-accurate regression function in the presence of higher poisoning rates would be an ill-posed task, if not infeasible at all~\cite{wang14-usenix,Huang2011adversarial,Xiao15,biggio17-aisec,Biggio2012poisoning,mei15-aaai}.}

\myparagraph{Poisoning Attack Strategy} \revision{All the aforementioned poisoning attack scenarios, encompassing availability and integrity violations under white-box or black-box knowledge assumptions, can be formalized as a \emph{bilevel optimization} problem~\cite{biggio17-aisec,mei15-aaai}.
For white-box attacks, this can be written as:}
\begin{eqnarray}
\label{eq:bilevel1}
\argmax_{\poisonset} & & \attackloss(\set D^\prime, \thetaoptpoison) \, , \\
\label{eq:bilevel2}
\mbox{s.t. } & & \thetaoptpoison \in \argmin_{\thetapar} \defenderloss({\trainset} \cup \poisonset, \thetapar) \, .
\end{eqnarray}
\revision{The outer optimization amounts to selecting the poisoning points $\poisonset$ to maximize a loss function $\attackloss$ on an untainted data set $\set D^\prime$ (\eg, a validation set which does not contain any poisoning points), while the inner optimization corresponds to retraining  the regression algorithm on a \emph{poisoned} training set including $\poisonset$.
It should be clear that $\thetaoptpoison$ depends \emph{implicitly} on the set $\poisonset$ of poisoning attack samples through the solution of the inner optimization problem.
In poisoning integrity attacks, the attacker's loss $\attackloss$ can be evaluated only on the points of interest (for which the attacker aims to cause mis-predictions at test time), while in poisoning availability attacks it is computed on an untainted set of data points, indiscriminately.
In the black-box setting, the poisoned regression parameters $\thetaoptpoison$ are estimated using the substitute training data $\trainsetsub$ instead of $\trainset$.}

In the remainder of this work, we only focus on poisoning availability attacks against regression learning, and on defending against them, as those have been mainly investigated in the literature of poisoning attacks. We highlight anyway again that poisoning integrity attacks can be implemented using the same technical derivation presented in this work, and leave a more detailed investigation of their effectiveness to future work.

\section{Attack methodology}
\label{sec:attacks}

In this section, we first discuss previously-proposed gradient-based optimization approaches to solving Problem~\eqref{eq:bilevel1}-\eqref{eq:bilevel2} in classification settings.
In Sect.~\ref{sec:opt-attack}, we discuss how to adapt them to the case of regression learning, and propose novel strategies to further improve their effectiveness. Notably, since these attacks have been originally proposed in the context of classification problems, the class label of the attack sample is arbitrarily initialized and then kept fixed during the optimization procedure (recall that $y$ is a categorical variable in classification). As we will demonstrate in the remainder of this work, a significant improvement we propose here to the current attack derivation is to \emph{simultaneously} optimize the response variable of each poisoning point along with its feature values.
We subsequently highlight some theoretical insights on how each poisoning sample is updated during the gradient-based optimization process.
This will lead us to develop a much faster attack, presented in Sect.~\ref{sec:stat}, which only leverages some statistical properties of the data and requires minimal black-box access to the targeted model.

\subsection{Optimization-based Poisoning Attacks} \label{sec:opt-attack}

Previous work has considered solving Problem~\eqref{eq:bilevel1}-\eqref{eq:bilevel2}
by iteratively optimizing one poisoning sample at a time through gradient ascent~\cite{Biggio2012poisoning,Xiao15,mei15-aaai,biggio17-aisec}.
An exemplary algorithm is given as Algorithm~\ref{alg:poisoning}.
We denote with $\vct x_c$ the feature vector of the attack point being optimized, and with $y_c$ its response variable (categorical for classification problems).
In particular, in each iteration, the algorithm optimizes all points in $\poisonset$,
by updating their feature vectors one at a time.
As reported in~\cite{Xiao15}, the vector $\vct x_c$ can be updated through a line search along the direction of the gradient $\nabla_{\vct x_c} \attackloss$ of the outer objective $\attackloss$ (evaluated at the current poisoned solution) with respect to the poisoning point $\vct x_c$ (cf. line 7 in Algorithm~\ref{alg:poisoning}). Note that this update step should also enforce $\vct x_c$ to lie within the feasible domain (\eg, $\vct x_c \in [0,1]^d$), which can be typically achieved through simple projection operators~\cite{Biggio2012poisoning,Xiao15,biggio17-aisec}.
The algorithm terminates when no sensible change in the outer objective $\attackloss$ is observed.

\begin{algorithm}[tb]
  \caption{Poisoning Attack Algorithm}
  \label{alg:poisoning}
\begin{flushleft}
  \textbf{Input:} $\set D = \trainset$ (white-box) or $\trainsetsub$ (black-box), $\set D^\prime$, $\defenderloss$, $\attackloss$, the initial poisoning attack samples $\poisonset^{(0)} = (\vct x_c, y_c)_{c=1}^{\npois}$, a small positive constant $\varepsilon$.
\end{flushleft}
  \begin{algorithmic}[1]
  \State{$i \leftarrow 0$ (iteration counter)}
  \State{$\thetapar^{(i)} \leftarrow \argmin_{\thetapar} \defenderloss({\set D} \cup \poisonset^{(i)}, \thetapar)$ }
		\Repeat
		\State{$w^{(i)} \leftarrow \attackloss(\set D^\prime, \thetapar^{(i)})$}
		\State{$\thetapar^{(i+1)} \leftarrow \thetapar^{(i)}$}
		\For {c = $1, \ldots, \npois$}
		\State{$\vct x_{c}^{(i+1)} \leftarrow \mbox{line\_search} \left ({\vct x_c}^{(i)} ,  \nabla_{{\vct x_c}} {\attackloss}(\set D^\prime, \thetapar^{(i+1)}) \right)$}
		\State{$\thetapar^{(i+1)} \leftarrow \argmin_{\thetapar} \defenderloss({\set D} \cup \poisonset^{(i+1)}, \thetapar)$ }
		\State{$w^{(i+1)} \leftarrow \attackloss(\set D^\prime, \thetapar^{(i+1)})$}
		\EndFor
		\State{$i \leftarrow i + 1$}
		\Until{$|w^{(i)} - w^{(i-1)}| < \varepsilon$}
		
  \end{algorithmic}
  \begin{flushleft}
  \textbf{Output:} the final poisoning attack samples $\poisonset \leftarrow \poisonset^{(i)}$
  \end{flushleft}

\end{algorithm}

\myparagraph{Gradient Computation} The aforementioned algorithm is essentially a standard gradient-ascent algorithm with line search. The challenging part is understanding how to compute the required gradient $\nabla_{\vct x_c} \attackloss (\set D^\prime, \thetapar)$, as this has to capture the implicit dependency of the parameters  $\thetapar$ of the inner problem on the poisoning point $\vct x_c$.
Indeed, assuming that $\attackloss$ does not depend directly on $\vct x_c$, but only through $\thetapar$, we can compute $\nabla_{\vct x_c} \attackloss(\set D^\prime, \thetapar)$ using the chain rule as:
\begin{equation}
\nabla_{\vct x_c} \attackloss = \nabla_{\vct x_c} \thetapar(\vct x_c)^\T \cdot \nabla_{\thetapar} \attackloss \, ,
\label{eq:grad-full}
\end{equation}
where we have made explicit that $\thetapar$ \emph{depends} on $\vct x_c$.
While the second term is simply the derivative of the outer objective with respect to the regression parameters, the first one captures the dependency of the solution $\thetapar$ of the learning problem on $\vct x_c$.

We focus now on the computation of the term $\nabla_{\vct x_c} \thetapar(\vct x_c)$.
While for bilevel optimization problems in which the inner problem is not convex (\eg, when the learning algorithm is a neural network) this requires efficient numerical approximations~\cite{biggio17-aisec}, when the inner learning problem is convex, the gradient of interest can be computed in closed form.
The underlying trick is to replace the inner learning problem (Eq.~\ref{eq:bilevel2}) with its Karush-Kuhn-Tucker (KKT) equilibrium conditions, \ie, $\nabla_{\thetapar} \defenderloss( \trainsetsub \cup \poisonset, \thetapar ) = \vct 0$, and require such conditions to remain valid while updating $\vct x_c$~\cite{Biggio2012poisoning,Xiao15,mei15-aaai,biggio17-aisec}.
 To this end, we simply impose that their derivative with respect to $\vct x_c$ remains at equilibrium, \ie,
 $\nabla_{\vct x_c} \left (\nabla_{\thetapar} \defenderloss(\trainsetsub \cup \poisonset, \thetapar) \right ) = \vct 0$. Now, it is clear that the function $\defenderloss$ depends explicitly on $\vct x_c$ in its first argument, and implicitly through the regression parameters $\thetapar$.
 Thus, differentiating again with the chain rule, one yields the following linear system:
 \begin{equation}
 \nabla_{\vct x_c} \nabla_{\thetapar} \defenderloss +  \nabla_{\vct x_c} \thetapar^\T \cdot \nabla^2_{\thetapar} \defenderloss = \vct 0 \, .
 \end{equation}
Finally, solving for $\nabla_{\vct x_c} \thetapar$, one yields:
 \begin{equation}
  \nabla_{\vct x_c} \thetapar^\T = \left[ \begin{array}{cc}
\frac{\partial \wb}{\partial\xb_c}^\T &
\frac{\partial b}{\partial\xb_c}^\T
\end{array} \right] = - \nabla_{\vct x_c} \nabla_{\thetapar} \defenderloss \left (\nabla^2_{\thetapar} \defenderloss \right)^{-1} \, .
 \end{equation}
For the specific form of $\defenderloss$ given in Eq.~\eqref{eq:learning}, it is not difficult to see that the aforementioned derivative becomes equal to that reported in~\cite{Xiao15} (except for a factor of $2$ arising from a different definition of the quadratic loss).
\begin{equation}
 \nabla_{\vct x_c} \thetapar^\T= -
\frac{2}{n}
\left[ \begin{array}{cc}
\mathbf{M} &
\wb
\end{array} \right]
\left[ \begin{array}{cc}
\mat{\Sigma} +\lambda \vct{g} & \vct{\mu} \\
\vct{\mu}^\T & 1
\end{array} \right]^{-1} \, ,
\label{eq:grad-theta}
\end{equation}
where $\mathbf{\Sigma} = \frac{1}{n} \sum_{i} \xb_{i} \xb_{i}^{\T}$, $\mathbf{\mu} = \frac{1}{n} \sum_{i} \xb_{i}$, and $\mathbf{M} = \xb_{c} \wb^{\T}  + \left(f(\xb_{c})-y_{c} \right)\mathbb I_d$.
As in~\cite{Xiao15}, the term $\mathbf{g}$ is zero for OLS and LASSO, the identity matrix $\mathbb{I}_d$ for ridge regression, and $(1-\rho) \mathbb{I}_d$ for the elastic net.

\myparagraph{Objective Functions} In previous work, the main objective used for $\attackloss$ has been typically a loss function computed on an untainted validation set $\valset = \{( \vct x'_i, y'_i)\}_{i=1}^{m}$~\cite{Biggio2012poisoning,mei15-aaai,biggio17-aisec}. Notably, only Xiao~\etal~\cite{Xiao15} have used a regularized loss function computed on the training data (excluding the poisoning points) as a proxy to estimate the generalization error on unseen data. The rationale was to avoid the attacker to collect an additional set of points.
In our experiments, we consider both possibilities, always using the ${\rm MSE}$ as the loss function:
\begin{align}
\label{eq:wtr}
\attackloss_{\rm tr}(\trainset, \thetapar) &= \textstyle\frac{1}{n}\textstyle \sum_{i=1}^{n} \left( f(\xb_i, \thetapar)-y_i \right)^2 + \lambda \Omega(\wb) \, ,\\
\label{eq:wval}
\attackloss_{\rm val}(\valset, \thetapar) &= \textstyle\frac{1}{m}\textstyle \sum_{j=1}^{m} \left( f(\xb'_j, \thetapar)-y'_j \right)^2   \, .
\end{align}
The complete gradient $\nabla_{\vct x_c} \attackloss$ (Eq.~\ref{eq:grad-full}) for these two objectives can thus be computed by multiplying Eq.~\eqref{eq:grad-theta} respectively to:

\begin{align}
\label{eq:wtr-grad}
\nabla_{\thetapar} \attackloss_{\rm tr} &=
\left[ \begin{array}{cc}
\nabla_{\wb } \attackloss_{\rm tr} \\
\nabla_{b } \attackloss_{\rm tr}
 \end{array} \right]
 \\
 &= \left[ \begin{array}{cc}
 \textstyle \frac{2}{n} \sum_{i=1}^n (f(\xb_i) - y_i)  \xb_i+ \lambda \diff{\Omega}{\wb}^\T \\
  \textstyle \frac{2}{n} \sum_{i=1}^n (f(\xb_i) - y_i)
 \end{array} \right] \, ,\\
\label{eq:wval-grad}
\nabla_{\thetapar} \attackloss_{\rm val} &=
\left[ \begin{array}{cc}
\nabla_{\wb } \attackloss_{\rm val} \\
\nabla_{b } \attackloss_{\rm val}
 \end{array} \right]
 =
\left[ \begin{array}{cc}
 \textstyle \frac{2}{m} \sum_{j=1}^m (f(\xb_j) - y_j)  \xb_j \\
  \textstyle \frac{2}{m} \sum_{j=1}^m (f(\xb_j) - y_j)
 \end{array} \right] \, .
\end{align}

\myparagraph{Initialization strategies}
We discuss here how to select the initial set $\poisonset$ of poisoning points to be passed as input to the gradient-based optimization algorithm (Algorithm~\ref{alg:poisoning}).
Previous work on poisoning attacks has only dealt with classification problems~\cite{Biggio2012poisoning,Xiao15,mei15-aaai,biggio17-aisec}.
For this reason, the initialization strategy used in all previously-proposed approaches has been to randomly clone a subset of the training data and flip their labels.
Dealing with regression opens up different avenues.
We therefore consider two initialization strategies in this work.
In both cases, we still select a set of points at random from the training set $\trainset$, but then we set the new response value $y_c$ of each poisoning point in one of two ways: ($i$) setting $y_c=1-y$, and $(ii)$ setting $y_c={\rm round}(1-y)$, where ${\rm round}$ rounds to the nearest $0$ or $1$ value (recall that the response variables are in $[0,1]$). We call the first technique \emph{Inverse Flipping} (\invflip) and the second \emph{Boundary Flipping} (\bflip).
Worth remarking, we experimented with many techniques for selecting the feature values before running gradient descent, and found that surprisingly they do not have significant improvement over a simple uniform random choice. We thus report results only for the two aforementioned initialization strategies.

\myparagraph{Baseline Gradient Descent (\baseline) Attack} We are now in a position to define a baseline attack against which we will compare our improved attacks. In particular, as no poisoning attack has ever been considered in regression settings, we define as the baseline poisoning attack an adaptation from the attack by Xiao~\etal~\cite{Xiao15}.
In particular, as in Xiao~\etal~\cite{Xiao15}, we select $\attackloss_{\rm tr}$ as the outer objective.
To simulate label flips in the context of regression, we initialize the response variables of the poisoning points with the \invflip\ strategy.
We nevertheless test all the remaining three combinations of initialization strategies and outer objectives in our experiments.

\myparagraph{Response Variable Optimization}
This work is the first to consider poisoning attacks in regression settings.
Within this context, it is worth remarking that response variables take on continuous values rather than categorical ones.
Based on this observation, we propose here the first poisoning attack that jointly optimizes the feature values $\vct x_c$ of poisoning attacks \emph{and} their associated response variable $y_c$.
To this end, we extend the previous gradient-based attack by considering the optimization of $\vct z_c = (\vct x_c, y_c)$ instead of only considering $\vct x_c$. This means that all previous equations remain valid provided that we substitute $\nabla_{\vct z_c}$ to $\nabla_{\vct x_c}$.
This clearly requires expanding $\nabla_{\vct x_c} \thetapar$ by also considering derivatives with respect to $y_c$:
\begin{equation}
 \nabla_{\vct z_c} \thetapar=
 \left[ \begin{array}{cc}
\diff{\wb}{\xb_c} &  \diff{\wb}{y_c}\\
\diff{b}{\xb_c} & \diff{b}{y_c}
\end{array} \right] \, ,
 \end{equation}
and, accordingly, modify Eq.~\eqref{eq:grad-theta} as
\begin{equation}
 \nabla_{\vct z_c} \thetapar^\T= -
\frac{2}{n}
\left[ \begin{array}{cc}
\mathbf{M} &
\wb \\
- \vct x_c^\T & -1
\end{array} \right]
\left[ \begin{array}{cc}
\mat{\Sigma} +\lambda \vct{g} & \vct{\mu} \\
\vct{\mu}^\T & 1
\end{array} \right]^{-1} \, .
\label{eq:grad-theta-yc}
\end{equation}
 The derivatives given in Eqs.~\eqref{eq:wtr-grad}-\eqref{eq:wval-grad} remain clearly unchanged, and can be pre-multiplied by Eq.~\eqref{eq:grad-theta-yc} to obtain the final gradient $\nabla_{\vct z_c} \attackloss$.
 Algorithm~\ref{alg:poisoning} can still be used to implement this attack, provided that both $\vct x_c$ and $y_c$ are updated along the gradient $\nabla_{\vct z_c} \attackloss$ (cf.  Algorithm~\ref{alg:poisoning}, line 7).

\myparagraph{Theoretical Insights}
We discuss here some theoretical insights on the bilevel optimization of Eqs.~\eqref{eq:bilevel1}-\eqref{eq:bilevel2}, which will help us to derive the basis behind the statistical attack introduced in the next section.
To this end, let us first consider as the outer objective a non-regularized version of $\attackloss_{\rm tr}$, which can be obtained by setting $\lambda=0$ in Eq.~\eqref{eq:wtr}.
As we will see, in this case it is possible to compute simplified closed forms for the required gradients.
Let us further consider another objective denoted with $\attackloss^\prime_{\rm tr}$, which, instead of optimizing the loss, optimizes the difference in predictions from the original, unpoisoned model $\thetapar^\prime$:
\[
\attackloss^\prime_{\rm tr} = \textstyle \frac{1}{n}\sum_{i=1}^{n} ( f(\vct x_i, \thetapar) - f(\vct x_i, \thetapar^\prime))^2.
\]

In Appendix~\ref{app:att-th}, we show that $\attackloss{\rm tr}$ and $\attackloss^\prime_{\rm tr}$ are interchangeable for our bilevel optimization problem.  In particular, differentiating $\attackloss^\prime_{\rm tr}$ with respect to $\vct z_c=(\xb_c,y_c)$ gives:
\begin{align}
\diff{\attackloss^\prime_{\rm tr}}{\xb_c} &= \textstyle\frac{2}{n}(f(\xb_c, \thetapar)-f(\xb_c, \thetapar^\prime))(\wb_0-2\wb)^\T\\
\diff{\attackloss^\prime_{\rm tr}}{y_c} &= \textstyle \frac{2}{n}(f(\xb_c, \thetapar) - f(\xb_c, \thetapar^\prime)).
\end{align}
The update rules defined by these gradients have nice interpretation. We see that $\diff{\attackloss^\prime_{\rm tr}}{y_c}$ will update $y_c$ to be further away from the original line than it was in the previous iteration. This is intuitive, as a higher distance from the line will push the line further in that direction. The update for $\xb_c$ is slightly more difficult to understand, but by separating $(\wb_0-2\wb)$ into $(-\wb) + (\wb_0-\wb)$, we see that the $\xb_c$ value is being updated in two directions summed together. The first is perpendicularly away from the regression line (like the $y_c$ update step, the poison point should be as far as possible from the regression line). The second is parallel to the difference between the original regression line and the poisoned regression line (it should keep pushing in the direction it has been going). This gives us an intuition for how the poisoning points are being updated, and what an optimal poisoning point looks like.

\subsection{Statistical-based Poisoning Attack (\rmml)}
\label{sec:stat}

Motivated by the aforementioned theoretical insights, we design a fast statistical attack that produces poisoned points with similar distribution as the training data. In \rmml, we simply sample from a multivariate normal distribution with the mean and covariance estimated from the training data. Once we have generated these points, \revision{we round the feature values to the corners, exploiting the observation that the most effective poisoning points are near corners. Finally, we select the response variable's value at the boundary (either 0 or 1) to maximize the loss.}

Note that, importantly, the \rmml\ attack requires only black-box access to the model, as it needs to query the model to find the response variable (before performing the boundary rounding). It also needs minimal information to be able to sample points from the training set distribution. In particular, \rmml\ requires an estimate of the mean and co-variance of the training data. However, \rmml\ is agnostic to the exact regression algorithm, its parameters, and training set. Thus, it requires much less information on the training process than the optimization-based attacks.
It is significantly faster than optimization-based attacks, though slightly less effective.

\newtheorem{thm}{Theorem}

\section{Defense Algorithms}
\label{sec:defense}

In this section, we describe existing defense proposals against poisoning attacks, and explain why they may not be effective under adversarial corruption in the training data. Then we present a new approach called \trim, specifically designed to increase robustness against a range of poisoning attacks.

\begin{figure}[t!]
\centering
\vspace{-10pt}
\includegraphics[width=\columnwidth]{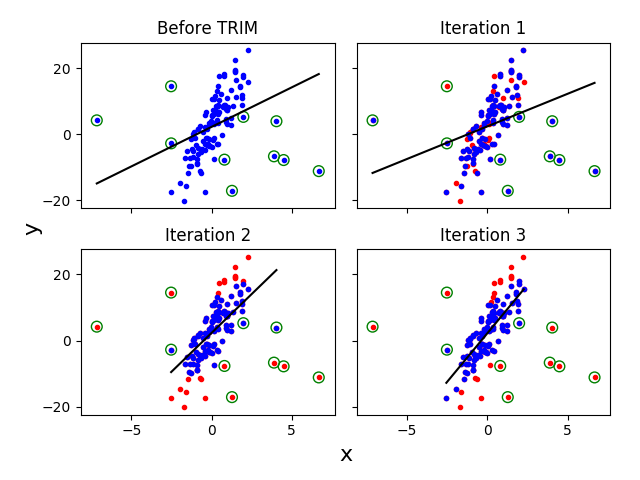}
\vspace{-10pt}
\caption{\revision{Several iterations of the \trim\ algorithm. Initial poisoned data is in blue in top left graph. The top right graph shows in red the initial randomly selected points removed from the optimization objective. In the following two iterations (bottom left and right graphs) the set of high-residual points is refined and the model becomes more robust.}}
\label{fig:trim}
\vspace{-15pt}
\end{figure}

\subsection{Existing defense proposals}

Existing defense proposals can be classified into two categories: noise-resilient regression algorithms and adversarially-resilient defenses. We discuss these approaches below.

\vspace{0.1cm}
\noindent {\bf Noise-resilient regression.} Robust regression has been extensively studied in statistics  as a method to provide resilience against noise and outliers~\cite{Huber64,tyler2008robust,Xu10,huber2011robust}. The main idea behind these approaches is to identify and remove \emph{outliers} from a dataset. For example, Huber~\cite{Huber64} uses an outlier-robust loss function. RANSAC~\cite{RANSAC} iteratively trains a model to fit a subset of samples selected at random, and identifies a training sample as an outlier if the error when fitting the model to the sample is higher than a threshold.

While these methods provide robustness guarantees against noise and outliers, an adversary can still generate poisoning data that affects the trained model. In particular, an attacker can generate poisoning points that are very similar to the true data distribution (these are called \emph{inliers}), but can still mislead the model. Our new attacks discussed in Section~\ref{sec:attacks} generate poisoning data points which are akin to the pristine ones. For example, in \statp\ the poisoned points are chosen from a distribution that is similar to that of the training data (has the same mean and co-variance). It turns out that these existing regression methods are not robust against inlier attack points chosen to maximally mislead the estimated regression model.

\vspace{0.1cm}
\noindent {\bf Adversarially-resilient regression.} Previously proposed adversarially-resilient regression algorithms typically provide guarantees under strong assumptions about data and noise distribution. For instance, Chen et al.~\cite{Chen13,Chen13sparse} assume that the feature set matrix satisfies $X^T X = I$ and data has sub-Gaussian distribution. Feng et al.~\cite{Feng14} assume that the data and noise satisfy the sub-Gaussian assumption. Liu et al.~\cite{Liu17-AISEC} design robust linear regression algorithms robust under the assumption that the feature matrix has low rank and can be projected to a lower dimensional space. All these methods have provable robustness guarantees, but the assumptions on which they rely are not usually satisfied in practice.

\begin{algorithm}[t]
    \begin{algorithmic}[1]

    \State {\bf Input}: Training data $\alltrain = \trainset\ \cup \poisonset$ with $|\alltrain| = N$; number of attack points $\npois = \alpha \cdot n$.
    \State {\bf Output}: $\thetapar$.

    \State $\mathcal{I}^{(0)}\leftarrow \{1,...,N\}$ /* First train with all samples */
    \State $\thetapar^{(0)}\leftarrow \argmin_{\thetapar} \defenderloss(\alltrain^{\mathcal{I}^{(0)}}, \thetapar)$ /* Initial estimation of $\thetapar$*/
    \State $i\leftarrow 0$ /* Iteration count */
    \Repeat
        \State $i\leftarrow i+1$;
        \State $\mathcal{I}^{(i)}\leftarrow$ subset of size $n$ that min. $\defenderloss(\alltrain^{\mathcal{I}^{(i)}},\thetapar^{(i-1)})$
        \State $\thetapar^{(i)} \leftarrow \argmin_{\thetapar} \defenderloss(\alltrain^{\mathcal{I}^{(i)}},\thetapar)$ /* Current estimator */
        \State $R^{(i)}=\defenderloss(\alltrain^{\mathcal{I}^{(i)}},\thetapar^{(i)})$ /* Current loss */
    \Until{$i>1\wedge R^{(i)}=R^{(i-1)}$} /* Convergence condition*/
    \State \Return $\thetapar^{(i)}$ /* Final estimator */.
    \caption{[\trim\ algorithm]}
    \label{alg:trim}
    \end{algorithmic}
\end{algorithm}

\subsection{\trim\ algorithm}

\revision{ In this section, we propose a novel defense algorithm called \trim\ with the goal of training a regression model with poisoned data. At an intuitive level, rather than simply removing outliers from the training set, \trim\ takes a principled approach. \trim\ iteratively estimates the regression parameters, while at the same time training on a subset of points with lowest residuals in each iteration. In essence, \trim\ uses a trimmed loss function computed on a different subset of residuals in each iteration. Our method is inspired by techniques from robust statistics that use trimmed versions of the loss function for robustness (e.g. \cite{rousseeuw2006computing}). Our main contribution is to apply trimmed optimization techniques for regularized linear regression in adversarial settings, and demonstrate their effectiveness compared to other defenses on a range of models and real-world datasets. }

As in Section~\ref{sec:model}, assume that the original training set is \trainset\ of size $n$, the attacker injects $\npois =  \alpha \cdot n $ poisoned samples \poisonset, and the poisoned training set $\alltrain = \trainset \cup \poisonset$ is of size $N =  (1+\alpha)n $. We require that $\alpha<1$ to ensure that the majority of training data is pristine (unpoisoned).

\revision{Our main observation is the following: we can train a linear regression model only using a subset of training points of size $n$. In the ideal case, we would like to identify all $\npois$ poisoning points and train the regression model based on the remaining $n$ legitimate points. However, the true distribution of the legitimate training data is clearly unknown, and it is thus difficult to separate legitimate and attack points precisely.
To alleviate this, our proposed defense tries to identify a set of training points with lowest residuals relative to the regression model (these might include attack points as well, but only those that are ``close'' to the legitimate points and do not contribute much to poisoning the model). In essence, our \trim\ algorithm provides a solution to the following optimization problem:}
\begin{equation}
\min_{\thetapar, \mathcal{I}} \defenderloss(\alltrain^\mathcal{I},\thetapar)\quad\mathrm{s.t.}~I \subset [1,\dots,N] \wedge |\mathcal{I}|=n\label{eq:opt} \, .
\end{equation}
We use the notation $\alltrain^{\mathcal{I}}$ to indicate the data samples $\{(\xb_i, y_i)\in \alltrain\}_{i\in\mathcal{I}}$. Thus, we optimize the parameter $\thetapar$ of the regression model and the subset $I$ of points with smallest residuals at the same time. \revision{It turns out though that solving this optimization problem efficiently is quite challenging. A simple algorithm that enumerates all subsets $I$ of size $n$ of the training set is computationally inefficient. On the other hand, if the true model parameters $\thetapar=(\wb, b)$ were known, then we could simply select points in set $I$ that have lowest residual relative to $\thetapar$. However, what makes this optimization problem difficult to solve is the fact that $\thetapar$ is not known, and we do not make any assumptions on the true data distribution or the attack points. }

To address these issues, our \trim\ algorithm learns parameter $\thetapar$ and distinguishes points with lowest residuals from training set alternately. We employ an iterative algorithm inspired by techniques such as alternating minimization or expectation maximization~\cite{AltMin}. At the beginning of iteration $i$, we have an estimate of parameter $\thetapar^{(i)}$. We use this estimate as a discriminator to identify all inliers, whose residual values are the $n$ smallest ones. We do not consider points with large residuals (as they increase MSE), but use only the inliers to estimate a new parameter $\thetapar^{(i+1)}$. This process terminates when the estimation converges and the loss function reaches a minimum. The detailed algorithm is presented in Algorithm~\ref{alg:trim}.  \revision{A graphical representation of three iterations of our algorithm is given in Figure~\ref{fig:trim}. As observed in the figure, the algorithm iteratively finds the direction of the regression model that fits the true data distribution, and identifies points that are outliers.}

We provide provable guarantees on the convergence of Algorithm \ref{alg:trim} and the estimation accuracy of the regression model it outputs. First, Algorithm~\ref{alg:trim} is guaranteed to converge and thus it terminates in finite number of iterations, as stated in the following theorem.

\vspace{-0.1cm}
\begin{restatable}{thm}{trimconv} Algorithm~\ref{alg:trim} terminates in a finite number of iterations.
\label{thm:1}
\end{restatable}
\vspace{-0.1cm}

\revision{We do not explicitly provide a bound on the number of iterations needed for convergence, but it is always upper bounded by ${N}\choose{n}$. However, our empirical evaluation demonstrates that Algorithm~\ref{alg:trim} converges within few dozens of iterations at most.}

We are next interested in analyzing the quality of the estimated model computed from Algorithm~\ref{alg:trim} (\emph{adversarial world}) and how it relates to the pristine data (\emph{ideal world}). However, relating these two models directly is challenging due to the iterative minimization used by Algorithm~\ref{alg:trim}. We overcome this by observing that Algorithm~\ref{alg:trim} finds a \emph{local minimum} to the optimization problem from (\ref{eq:opt}). \revision{There is no efficient algorithm for solving (\ref{eq:opt}) that guarantees the solution to be the global minimum of the optimization problem.}

It turns out that we can provide a guarantee about the \emph{global minimum} $\hat{\thetapar}$ of (\ref{eq:opt}) on poisoned data (under worst-case adversaries) in relation to the parameter \thetaopt\ learned by the original model on pristine data. In particular, Theorem~\ref{thm:2} shows that $\hat{\thetapar}$ ``fits well" to at least $(1-\alpha)\cdot n$ pristine data samples. Notably, it does not require any assumptions on how poisoned data is generated, thus it provides guarantees under worst-case adversaries.
\vspace{-0.1cm}

\begin{restatable}{thm}{trimloss} Let \trainset\ denote the original training data, $\hat{\thetapar}$ the global optimum for (\ref{eq:opt}), and $\thetaopt=\argmin_{\thetapar} \defenderloss(\trainset, \thetapar)$ the estimator in the ideal world on pristine data. Assuming $\alpha<1$, there exist a subset $\mathcal{D}' \subset \trainset$ of $(1-\alpha)\cdot n$ pristine data samples such that
\begin{equation}
\mathrm{MSE}(\mathcal{D}',\hat{\thetapar}) \leq \bigg( 1+\frac{\alpha}{1-\alpha}\bigg)  \defenderloss(\trainset,\thetaopt) \, .
\label{eq:guarantee}
\end{equation}
\label{thm:2}
\end{restatable}%
\vspace{-0.2cm}

\revision{Note that the above theorem is stated without any assumptions on the training data distribution. This is one of the main difference from prior work~\cite{Chen13, Feng14}, which assume the knowledge of the mean and covariance of the legitimate data. In practice, such information on training data is typically unavailable. Moreover, an adaptive attacker can also inject poisoning samples to modify the mean and covariance of training data. Thus, our results are stronger than prior work in relying on fewer assumptions.}

We now give an intuitive explanation about the above theorem, especially inequality~(\ref{eq:guarantee}). Since $\trainset$ is assumed to be the pristine dataset, and $\mathcal{D}'$ is a subset of $\trainset$ of size $(1-\alpha)n$, we know all data in $\mathcal{D}'$ is also pristine (not corrupted by the adversary). Therefore, the stationary assumption on pristine data distribution, which underpins all machine learning algorithms, guarantees that $\mathrm{MSE}(\trainset, \thetapar)$ is close to $\mathrm{MSE}( \mathcal{D}',\thetapar)$ regardless of the choices of $\thetapar$ and $\mathcal{D}'$, as long as $\alpha$ is small enough.

Next, we explain the left-hand side of inequality~(\ref{eq:guarantee}). This is the MSE of a subset of pristine samples $\mathcal{D}'$ using $\hat{\thetapar}$ computed by the \trim\ algorithm in the adversarial world. Based on the discussion above, the left-hand side is close to the MSE of the pristine data $\trainset$ using the adversarially learned estimator $\hat{\thetapar}$. Thus, inequality (\ref{eq:guarantee}) essentially provides an upper bound on the worst-case MSE using the estimator $\hat{\thetapar}$ output by Algorithm~\ref{alg:trim} from the poisoned data.

To understand what upper bound Theorem~\ref{thm:2} guarantees, we need to understand the right-hand side of inequality (\ref{eq:guarantee}). We use OLS regression (without regularization) as an example to explain the intuition of the right-hand side. In OLS we have $\defenderloss(\trainset, \thetaopt)=\mathrm{MSE}(\trainset,\thetaopt)$, which is the MSE using the ``best-case" estimator computed in the ideal world. Therefore, the right-hand side of inequality (\ref{eq:guarantee}) is proportional to the ideal world MSE, with a factor of $(1+\frac{\alpha}{1-\alpha})$. When $\alpha\leq 20\%$, we notice that this factor is at most $1.25\times$.

\revision{Therefore, informally, Theorem~\ref{thm:2} essentially guarantees that, the ratio of the worst-case MSE by solving (\ref{eq:guarantee}) computed in the adversarial world over best-case MSE computed in ideal world for a linear model is at most $1.25$. Note that since Algorithm~\ref{alg:trim} may not always find the global minimum of (\ref{eq:opt}), we empirically examine this ratio of the worst-case to best-case MSEs. Our empirical evaluation shows that in most of our experiments, this ratio for \trim\ is less than for existing defenses. At $\alpha=12\%$, it is always below $1.46 \times$, and reaches $2.39\times$ for one dataset at $\alpha=20\%$. This ratio depends significantly on the initial indices $\mathcal{I}^{(0)}$. A reasonable initialization strategy to find a good local optimum is to train on all $N$ points initially.}

\revision{For other models whose loss function includes the regularizer term (Lasso, ridge, and elastic net), the right-hand side of (\ref{eq:guarantee}) includes the same term as well. This may allow the blowup of the worst-case MSE in the adversarial world with respect to the best-case MSE to be larger; however, we are not aware of any technique to trigger this worst-case scenario.}

The proofs of Theorem~\ref{thm:1} and~\ref{thm:2} can be found in Appendix~\ref{app:analysis}.

\section{Experimental evaluation}
\label{sec:eval}

We implemented our attack and defense algorithms in Python, using the numpy and sklearn packages. Our code is available at \url{https://github.com/jagielski/manip-ml}. We ran our experiments on four 32 core Intel(R) Xeon(R) CPU E5-2440 v2 @ 1.90GHz machines. We parallelize our  optimization-based attack implementations to take advantage of the multi-core capabilities. We use the standard cross-validation method to split the datasets into 1/3 for training, 1/3 for testing, and 1/3 for validation, and report results as averages over 5 runs. We use two main metrics for evaluating our algorithms: MSE for the effectiveness of the attacks and defenses, and running time for their cost.

We describe the datasets we used for our experiments in Section~\ref{sec:dataset}. We then systematically analyze the performance of the new attacks and compare them against the baseline attack algorithm in Section~\ref{sec:exp-new-attacks}. Finally, we present the results of our new \trim\ algorithm and compare it with previous methods from robust statistics in Section~\ref{sec:exp-defense}.

\subsection{Datasets}
\label{sec:dataset}

We used three public regression datasets in our experimental evaluation. We present some details and statistics about each of them below.

\vspace{0.1cm}
\noindent {\bf Health care dataset.} This dataset includes 5700 patients, where the goal is to predict the dosage of anticoagulant drug Warfarin using demographic information, indication for Warfarin use, individual VKORC1 and CYP2C9 genotypic data, and use of other medications affected by related VKORC1 and CYP2C9 polymorphisms \cite{pharmagkb}. As is standard practice for studies using this dataset (see \cite{fredrikson2014privacy}), we only select patients with INR values between 2 and 3. The INR is a ratio that represents the amount of time it takes for blood to clot, with a therapeutic range of 2-3 for most patients taking Warfarin. The dataset includes 67 features, resulting in 167 features after one-hot encoding categorical features and normalizing numerical features as above.

\vspace{0.1cm}
\noindent {\bf Loan dataset.} This dataset contains information regarding loans made on the Lending Club peer-to-peer lending platform \cite{lendingclub}. The predictor variables describe the loan attributes, including information such as total loan size, interest rate, and amount of principal paid off, as well as the borrower's information, such as number of lines of credit, and state of residence. The response variable is the interest rate of a loan. Categorical features, such as the purpose of the loan, are one-hot encoded, and numerical features are normalized into [0,1]. The dataset contains 887,383 loans, with 75 features before pre-processing, and 89 after. Due to its large scale, we sampled a set of 5000 records for our poisoning attacks.

\vspace{0.1cm}
\noindent {\bf House pricing dataset.}  This dataset is used to predict house sale prices as a function of predictor variables such as square footage, number of rooms, and location~\cite{houseprices}. In total, it includes 1460 houses and 81 features. We preprocess by one-hot encoding all categorical features and normalize numerical features, resulting in 275 total features.

\subsection{New poisoning attacks}
\label{sec:exp-new-attacks}

\begin{figure*}[tbhp]
\centering
\begin{subfigure}[t]{0.66\columnwidth}
\includegraphics[width=\columnwidth]{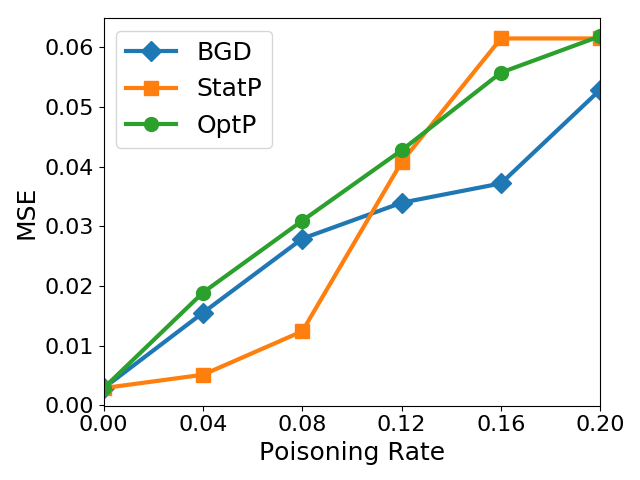}
\caption{Health Care Dataset}
\label{fig:health-ridge-newattacks}
\end{subfigure}
\begin{subfigure}[t]{0.66\columnwidth}
\includegraphics[width=\columnwidth]{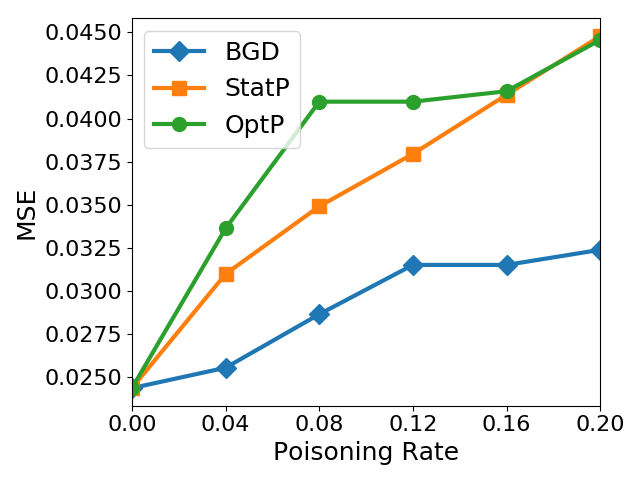}
\caption{Loan Dataset}
\label{fig:loan-ridge-newattacks}
\end{subfigure}
\begin{subfigure}[t]{0.66\columnwidth}
\includegraphics[width=\columnwidth]{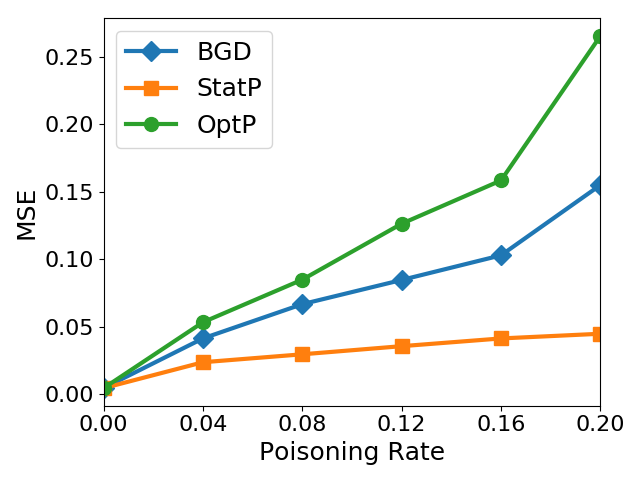}
\caption{House Price Dataset}
\label{fig:house-ridge-newattacks}
\end{subfigure}
\caption{\revision{MSE of attacks on ridge regression on the three datasets. Our new optimization (\optp) and statistical (\statp) attacks  are more effective than the baseline. \optp\ is best optimization attack according to Table~\ref{tab:optpoison}.}}
\label{fig:newattacks-ridge}
\vspace{-0.2cm}
\end{figure*}

\begin{figure*}[htbp]
\centering
\begin{subfigure}[t]{0.66\columnwidth}
\includegraphics[width=\columnwidth]{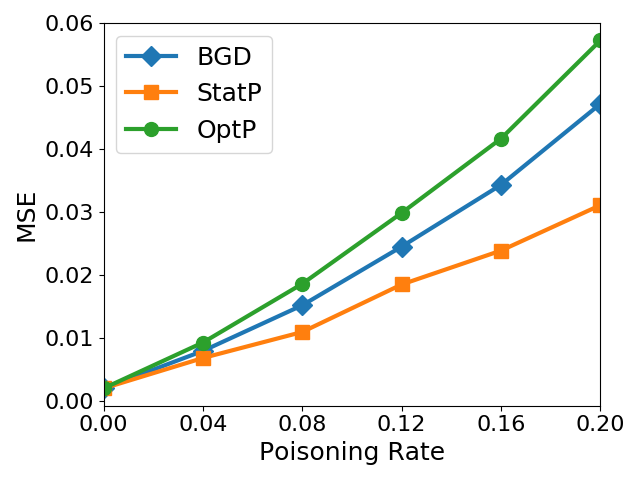}
\caption{Health Care Dataset}
\label{fig:health-lasso-newattacks}
\end{subfigure}
\begin{subfigure}[t]{0.66\columnwidth}
\includegraphics[width=\columnwidth]{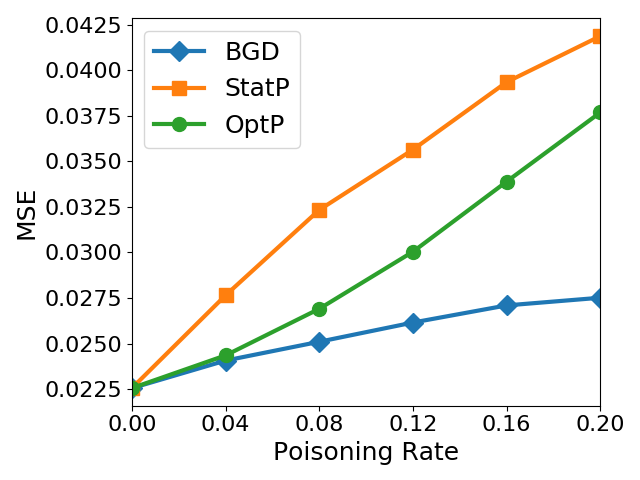}
\caption{Loan Dataset}
\label{fig:loan-lasso-newattacks}
\end{subfigure}
\begin{subfigure}[t]{0.66\columnwidth}
\includegraphics[width=\columnwidth]{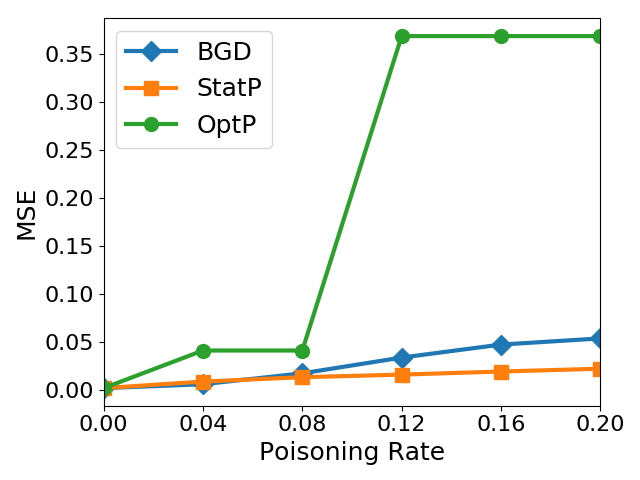}
\caption{House Price Dataset}
\label{fig:house-lasso-newattacks}
\end{subfigure}
\caption{\revision{MSE of attacks on LASSO on the three datasets. As for ridge, we find that \statp\ and \optp\ are able to poison the dataset very effectively, outperforming the baseline (\baseline). \optp\ is best optimization attack according to Table~\ref{tab:optpoison}.}}
\label{fig:newattacks-lasso}
\vspace{-0.2cm}
\end{figure*}

In this section, we perform experiments on the three regression datasets (health care, loan, and house pricing) to evaluate the newly proposed attacks, and compare them against the baseline \baseline\ \cite{Xiao15} for four regression models. For each dataset we select a subset of 1400 records (this is the size of the house dataset, and we wanted to use the same number of records for all datasets).  We use MSE as the metric for assessing the effectiveness of an attack, and also measure the attacks' running times. We vary the poisoning rate between 4\% and 20\% at intervals of 4\% with the goal of inferring the trend in attack success. More details about hyperparameter setting are presented in Appendix~\ref{sec:optattacks}.

Figures~\ref{fig:newattacks-ridge} and \ref{fig:newattacks-lasso} show the MSE of each attack for ridge and LASSO regression. We picked these two models as they are the most popular linear regression models. We plot the baseline attack \baseline, statistical attack \statp, as well as our best performing optimization attack (called \optp). Details on \optp\ are given in Table~\ref{tab:optpoison}. \revision{Additional results for the Contagio PDF classification dataset are given in Appendix~\ref{sec:optattacks}.}

Below we pose several research questions to elucidate the benefits, and in some cases limitations, of these attacks.

\subsubsection{Question 1: Which optimization strategies are most effective for poisoning regression?}

\revision{
Our results confirm that the optimization framework we design is effective at poisoning different models and datasets. Our new optimization attack \optp\ improves upon the baseline \baseline\ attack by a factor of 6.83 in the best case. The \optp\ attack could achieve MSEs by a factor of 155.7 higher than the original models.
}

\begin{table}
\scriptsize
\begin{center}
\begin{tabular}{|c|c|c|c|c|}
\hline

Model & Dataset & Init & Argument & Objective \\
\hline
Ridge & Health & \bflip\ & $(x,y)$ & \trainobj \\
& Loan & \bflip\ & $x$ & \valobj \\
& House & \bflip & $(x,y)$ & \trainobj \\
\hline
LASSO & Health & \bflip\ & $(x,y)$ & \trainobj \\
& Loan & \bflip\ & $(x,y)$ & \valobj \\
& House & \invflip & $(x,y)$ & \valobj \\
\hline
\end{tabular}
\end{center}
\caption{\revision{Best performing optimization attack \optp\ for Ridge and LASSO regression.} }
\label{tab:optpoison}
\end{table}

\revision{
As discussed in Section~\ref{sec:attacks}, our optimization framework has several instantiations, depending on: (1) The initialization strategy (\invflip\ or \bflip); (2) The optimization variable ($x$ or $(x,y)$); and (3) The objective of the optimization (\trainobj\ or \valobj). For instance, \baseline\ is given by (\invflip, $x$, \trainobj). We show that each of these dimensions has an important effect in generating successful attacks. Table~\ref{tab:optpoison} shows the best optimization attack for each model and dataset, while Tables~\ref{tab:lassoloan} and \ref{tab:lassohouse} provide examples of different optimization attacks for LASSO on the loan and house datasets, respectively.}

\revision{
We highlight several interesting observations. First, boundary flip \bflip\ is the preferred initialization method, with only one case (LASSO regression on house dataset) in which \invflip\ performs better in combination with optimizing $(x,y)$ under objective \valobj. For instance, in LASSO on house dataset, \bflip\ alone can achieve a factor of 3.18 higher MSE than \baseline\ using \invflip. In some cases the optimization by $y$ can achieve higher MSEs even starting with non-optimal $y$ values as the gradient ascent procedure is very effective (see for example the attack $(\invflip, (x,y), \valobj)$ in Table~\ref{tab:lassohouse}). However, the combination of optimization by $x$ with \invflip\ initialization (as used by \baseline) is outperformed in all cases by either \bflip\ or $(x,y)$ optimization.}

\revision{Second, using both $(x,y)$ as optimization arguments is most effective compared to simply optimizing by $x$ as in \baseline. Due to the continuous response variables in regression, optimizing by $y$ plays a large role in making the attacks more effective. For instance, optimizing by $(x,y)$ with \bflip\ initialization and \valobj\ achieves a factor of 6.83 improvement in MSE compared to \baseline\ on house dataset with LASSO regression.}

\revision{
Third, the choice of the optimization objective is equally important for each dataset and model. \valobj\ can improve over \trainobj\ by a factor of 7.09  (on house for LASSO), by 17.5\% (on loan for LASSO), and by 30.4\% (on loan for ridge) when the initialization points and optimization arguments are the same.}

\revision{
Thus, all three dimensions in our optimization framework are influential in improving the success of the attack. The optimal choices are dependent on the data distribution, such as feature types, sparsity of the data, ratio of records over data dimension, and data linearity. In particular, we noticed that for non-linear datasets (such as loan), the original MSE is already high before the attack  and  all the attacks that we tested perform worse than in cases when the legitimate data fits a linear model (\ie, it is close to the regression hyperplane).
The reason may be that, in the latter case, poisoning samples may be shifted farther away from the legitimate data (\ie, from the regression hyperplane), and thus have a greater impact than in the former case, when the legitimate data is already more evenly and non-linearly distributed in feature space.
Nevertheless, our attacks are able to successfully poison a range of models and datasets.
}

\begin{table}
\scriptsize
\begin{center}
\vspace{-0.2cm}
\begin{tabular}{|c|c|c|c|c|c|}
\hline

Init & Argument & Objective & \multicolumn{3}{|c|}{Poisoning rates} \\
\cline{4-6}
& & & 12\% & 16\% & 20\% \\
\hline
\invflip & $x$ & \trainobj\ & 0.026 &	0.027 &	0.027 \\
\bflip & $x$  & \trainobj\  &	0.028 &	0.032 &	0.033 \\
\invflip\ & $(x,y)$ &	\trainobj\ &	0.026 & 0.027 &	0.029 \\
\bflip\ & $(x,y)$	& \trainobj\ & 0.029 &	0.0316 & 0.032 \\
\bflip\ & $(x,y)$	& \valobj\	&	0.030 &	0.0338	& 0.0376 \\
\hline
\end{tabular}
\end{center}
\caption{\revision{MSEs of optimization attacks for LASSO on loan data. \baseline\ is the first row.}}
\vspace{-0.2cm}
\label{tab:lassoloan}
\end{table}

\begin{table}
\scriptsize
\begin{center}
\begin{tabular}{|c|c|c|c|c|c|}
\hline

Init & Argument & Objective & \multicolumn{3}{|c|}{Poisoning rates} \\
\cline{4-6}
& & & 12\% & 16\% & 20\% \\
\hline

\invflip & $x$ & \trainobj\ & 	0.034 &	0.047 &	0.054 \\
\bflip\ & $x$ & \trainobj\		& 0.08 	& 0.145	& 0.172 \\
\invflip\ & $(x,y)$	& \trainobj\	&	0.04 &	0.047 &	0.052 \\
\invflip\ & $(x,y)$	 & \valobj\ & 0.369 &	0.369 &	0.369 \\
\bflip\ & $(x,y)$	 & \trainobj\ &	0.08 &	0.145 &	0.172 \\
\hline
\end{tabular}
\end{center}
\caption{\revision{MSEs of optimization attacks for LASSO on house data. \baseline\ is the first row.}}
\vspace{-0.5cm}
\label{tab:lassohouse}
\end{table}

\subsubsection{Question 2: How do optimization and statistical attacks compare in effectiveness and performance?}

\revision{
In general, optimization-based attacks (\baseline\ and \optp) outperform the statistical-based attack \statp\ in effectiveness. This is not surprising to us, as \statp\ uses much less information about the training process to determine the attack points. Interestingly, we have one case (LASSO regression on loan dataset) in which \statp\ outperforms the best optimization attack \optp\ by 11\%. There are also two instances on ridge regression (health and loan datasets) in which \statp\ and \optp\ perform similarly. These cases show that \statp\ is a reasonable attack when the attacker has limited knowledge about the learning system.}

The running time of optimization attacks is proportional to the number of iterations required for convergence. \revision{On the highest-dimensional dataset, house prices, we observe \optp\ taking about 337 seconds to complete for ridge and 408 seconds for LASSO. On the loan dataset, \optp\ finishes LASSO poisoning in 160 seconds on average.  }As expected, the statistical attack is extremely fast, with running times on the order of a tenth of a second on the house dataset and a hundredth of a second on the loan dataset to generate the same number of points as \optp. Therefore, our attacks exhibit clear tradeoffs between effectiveness and running times, with optimization attacks being more effective than statistical attacks, at the expense of higher computational overhead.

\subsubsection{Question 3: What is the potential damage of poisoning in real applications?}

We are interested in understanding the effect of poisoning attacks in real applications, and perform a case study on the health-care dataset. Specifically, we translate the MSE results obtained with our attacks into application specific parameters. In the health care application, the goal is to predict  medicine dosage for the anticoagulant drug Warfarin. In Table~\ref{tab:dosage}, we show first statistics on the medicine dosage predicted by the original regression models (without poisoning), and then the absolute difference in the amount of dosage prescribed after the \optp\ poisoning attack. We find that all linear regression models are vulnerable to poisoning, with 75\% of patients having their dosage changed by \revision{93.49\%}, and half of patients having their dosage changed by \revision{139.31\%} on LASSO. For 10\% of patients, the increase in MSE is devastating to a maximum of \revision{359\%} achieved for LASSO regression. These results are for 20\% poisoning rate, but it turns out that the attacks are also effective at smaller poisoning rates. For instance, at 8\% poisoning rate, the change in dosage is \revision{75.06\%} for half of patients.

Thus, the results demonstrate the effectiveness of our new poisoning attacks that induce significant changes to the dosage of most patients with a small percentage of poisoned points added by the attacker.

\subsubsection{Question 4: What are the transferability properties of our attacks?}
Our transferability analysis for poisoning attacks is based on the black-box scenario discussed in Sect.~\ref{sec:model}, in which the attacker uses a substitute training set $\trainsetsub$ to craft the poisoning samples, and then tests them against the targeted model (trained on $\trainset$).
Our results, averaged on 5 runs, are detailed in Table~\ref{tab:surrogate}, which presents the ratio between transferred and original attacks. Note that the effectiveness of transferred attacks is very similar to that of the original attacks, with some outliers on the house dataset. For instance, the statistical attack \statp\ achieves transferred MSEs within 11.4\% of the original ones. \revision{The transferred \optp\ attacks have lower MSEs by 3\% than the original attack on LASSO. At the same time, transferred attacks could also improve the effectiveness of the original attacks: by 30\% for ridge, and 78\% for LASSO.}
We conclude that, interestingly, our most effective poisoning attacks (\optp\ and \statp) tend to have good transferability properties. There are some exceptions (ridge on house dataset), which deserve further investigation in future work. In most cases the MSEs obtained when using a different training set for both attacks is comparable to MSEs obtained when the attack is mounted on the actual training set.

\vspace{0.1cm}
\noindent{\bf Summary of poisoning attack results.}
\begin{itemize}

\item  \revision{We introduce a new optimization framework for poisoning regression models, which manages to improve upon \baseline\ by a factor of 6.83. The best attack \optp\ selects the initialization strategy, optimization argument, and objective to achieve maximum MSEs. }

\item We find that our statistical-based attack (\statp) works reasonably well in poisoning all datasets and models, is efficient in running time, and needs minimal information on the model. Our optimization-based attack \optp\ takes longer to run, needs more information on the model, but can be more effective in poisoning than \statp\ if properly configured.

\item In a health care case study, we find that our \optp\ attack can cause half of patients' Warfarin dosages to change by an average of \revision{139.31\%}. One tenth of these patients can have their dosages changed by \revision{359\%}, demonstrating the devastating consequences of poisoning.

\item We find that both our statistical and optimization attacks have good transferability properties, and still perform well with minimal difference in accuracy, when applied to different training sets.

\end{itemize}

\begin{table}
\scriptsize
\begin{center}
\begin{tabular}{|c|c|c|c|c|c|}
\hline
Quantile & Initial & Ridge & LASSO  \\
& Dosage & Diff & Diff\\
\hline

0.1 & 15.5 & 31.54\% & 37.20\% \\
0.25 & 21 & 87.50\%  & 93.49\% \\
0.5 & 30 & 150.99\% & 139.31\% \\
0.75 & 41.53 & 274.18\% & 224.08\% \\
0.9 & 52.5 & 459.63\% &  358.89\% \\
\hline
\end{tabular}
\end{center}
\caption{\revision{Initial dosage distribution (mg/wk) and percentage difference between original and predicted dosage after \optp\ attack at 20\% poisoning rate (health care dataset)}. }
\label{tab:dosage}
\end{table}

\begin{table}
\scriptsize
\begin{center}
\begin{tabular}{|c|c|c|c|c|c|}
\hline
Dataset & Attack & LASSO & Ridge\\
\hline

Health & \optp & 1.092 & 1.301  \\
 & \rmml & 0.971 & 0.927 \\
\hline
Loan & \optp & 1.028 & 1.100  \\
 & \rmml & 1.110 & 0.989 \\
\hline
House & \optp & 1.779 & 0.479  \\
 & \rmml & 1.034 & 0.886  \\
\hline
\end{tabular}
\end{center}
\caption{\revision{Transferability of \optp\ and \statp\ attacks. Presented are the ratio of the MSE obtained with transferred attacks over original attacks. Values below 1 represent original attacks outperforming transferred attacks, while values above 1 represent transferred attacks outperforming original attacks.}}
\label{tab:surrogate}
\end{table}

\subsection{Defense algorithms}
\label{sec:exp-defense}

\begin{figure*}[tbhp]
\centering
\vspace{-0.3cm}
\begin{subfigure}[t]{0.66\columnwidth}
\includegraphics[width=\columnwidth]{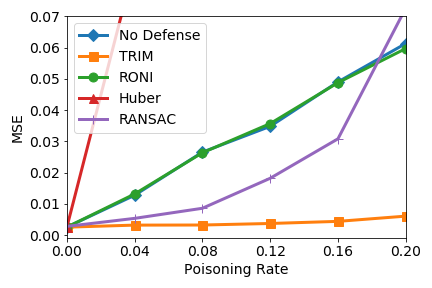}
\caption{Health Care Dataset}
\end{subfigure}
\begin{subfigure}[t]{0.66\columnwidth}
\includegraphics[width=\columnwidth]{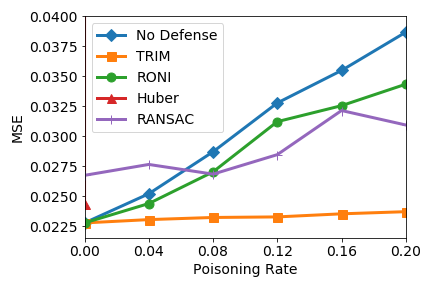}
\caption{Loan Dataset}
\end{subfigure}
\begin{subfigure}[t]{0.66\columnwidth}
\includegraphics[width=\columnwidth]{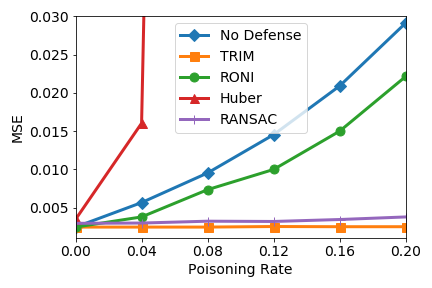}
\caption{House Price Dataset}
\end{subfigure}
\caption{\revision{MSE of defenses on ridge on the three datasets.  We exclude Chen from the graphs due to its large variability. Defenses are evaluated against the \optp\ attack. The only defense that consistently performs well in these situations is our proposed \trim\ defense, with RANSAC, Huber, and RONI actually performing worse than the undefended model in some cases. }}
\label{fig:defenseridge}
\vspace{-0.2cm}
\end{figure*}

\begin{figure*}[htbp]
\centering
\begin{subfigure}[t]{0.66\columnwidth}
\includegraphics[width=\columnwidth]{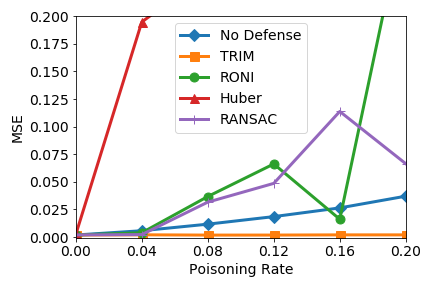}
\caption{Health Care Dataset}
\end{subfigure}
\begin{subfigure}[t]{0.66\columnwidth}
\includegraphics[width=\columnwidth]{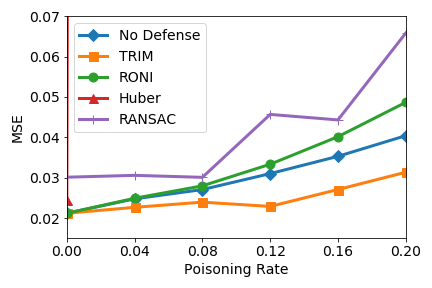}
\caption{Loan Dataset}
\end{subfigure}
\begin{subfigure}[t]{0.66\columnwidth}
\includegraphics[width=\columnwidth]{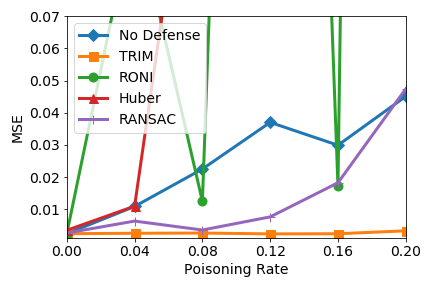}
\caption{House Price Dataset}
\end{subfigure}
\caption{\revision{MSE of defenses on LASSO. We exclude Chen from the graphs due to its large variability. Defenses are evaluated against the most effective attack \optp. As with ridge, the only defense that consistently performs well is our \trim\ defense. }}
\label{fig:defenselasso}
\vspace{-0.55cm}
\end{figure*}

In this section, we evaluate our proposed \trim\ defense and other existing defenses from the literature (Huber, RANSAC, Chen, \revision{ and RONI}) against the best performing optimization attacks from the previous section (\optp). We test two well-known methods from robust statistics: Huber regression~\cite{Huber64} and \RANSAC~\cite{RANSAC}, available as implementations in Python's sklearn package. Huber regression modifies the loss function from the standard MSE to reduce the impact of outliers. It does this by using quadratic terms in the loss function for points with small residuals and linear terms for points with large residuals. \revision{The threshold where linear terms start being used is tuned by a parameter $\epsilon>1$, which we set by selecting the best of 5 different values: $\lbrace 1.1,1.25,1.35, 1.5, 2\rbrace$.} RANSAC builds a model on a random sample of the dataset, and computes the number of points that are outliers from that model. If there are too many outliers, the model is rejected and a new model is computed on a different random dataset sample. \revision{The size of the initial random sample is a parameter that requires tuning - we select 5 different values, linearly interpolating from 25 to the total number of clean data, and select the value which has the lowest MSE. If the number of outliers is smaller than the number of poisoning points, we retain the model.}

We also compare against our own implementation of the robust regression method by Chen et al.~\cite{Chen13} from the machine learning community, \revision{and the RONI method from the security community~\cite{Nelson08}}. Chen picks the features of highest influence using an outlier resilient dot product computation. We vary the number of features selected by Chen (the only parameter in the algorithm) between 1 and 9 and pick the best results. We find that Chen has highly variable performance, having MSE increases of up to a factor of \revision{63,087} over the no defense models, and we decided to not include it in our graphs. The poor performance of Chen is due to the strong assumptions of the technique (sub-Gaussian feature distribution and covariance matrix $\Xb^T \Xb= \mathbb{I}$.), that are not met by our real world datasets. \revision{While we were able to remove the assumption that all features had unit variance through robust scaling (using the robust dot product provided by their work), removing the covariance terms would require a robust matrix inversion, which we consider beyond the scope of our work.}

\revision{RONI (Reject On Negative Impact) was proposed in the context of spam filters and attempts to identify outliers by observing the performance of a model trained with and without each point. If the performance degrades too much on a sampled validation set (which may itself contain outliers), the point is identified as an outlier and not included in the model. This method has some success in the spam scenario due to the ability of an adversary to send a single spam email with all words in dictionary, but is not applicable in other settings in which the impact of each point is small. We set the size of the validation set to 50, and pick the best points on average from 5 trials, as in the original paper. The size of the training dataset is selected from the same values as RANSAC's initial sample size.}

We show in Figures~\ref{fig:defenseridge} and \ref{fig:defenselasso} MSEs for ridge and LASSO regression for the original model (no defense), the \trim\ algorithm, as well as the Huber, RANSAC, \revision{and RONI} methods. We pose three research questions next:

\subsubsection{Question 1: Are known methods effective at defending against poisoning attacks?} \revision{As seen in Figures~\ref{fig:defenseridge} and~\ref{fig:defenselasso}, existing  techniques (Huber regression, RANSAC, and RONI), are not consistently effective at defending against our presented attacks. Huber frequently performs worse than no defense at all. For LASSO models, the \optp\ attack increases MSE over unpoisoned models by a factor of 19.74 (on the health care dataset at $\alpha=20\%$). Rather than decreasing the MSE, Huber regression in fact increases the MSE over undefended ridge models by a factor of 2781. RONI also increases the MSE of undefended models by a factor of 8.06\%. RANSAC performs the best, only increasing MSE by a factor of 1.79. The reason for this poor performance is that robust statistics methods are designed to remove or reduce the effect of outliers from the data, while RONI can only identify outliers with high impact on the trained models. Our attacks generate inlier points that have similar distribution as the training data, making these previous defenses ineffective.}

\subsubsection{Question 2: What is the robustness of the new defense \trim\ compared to known methods?} \revision{Our \trim\ technique is much more effective at defending against all attacks than the existing techniques are. For ridge and LASSO regression, the median MSE increase over all datasets and attacks is only 6.1\%, and only 20\% of attacks cause more than 27.2\% MSE increase. Interestingly, on the health care dataset the MSE of \trim\ is lower by 3.47\% compared to unpoisoned models for LASSO regression at $\alpha=8\%$. \trim\ achieves MSEs much lower than existing methods, improving Huber by a factor of 131.8, RANSAC by a factor of 17.5, and RONI by a factor of 20.28. This demonstrates that the \trim\ technique is a significant improvement over prior work at defending against these poisoning attacks.}

\subsubsection{Question 3: What is the running time of various defense algorithms?}
\revision{All of the defenses we evaluated ran in a reasonable amount of time, but \trim\ is the fastest. For example, on the house dataset, \trim\ took an average of 0.02 seconds, RANSAC took an average of 0.33 seconds, Huber took an average of 7.86 seconds, RONI took an average of 15.69 seconds and Chen took an average of 0.83 seconds. On the health care dataset, \trim\ took an average of 0.02 seconds, RANSAC took an average of 0.30 seconds, Huber took an average of 0.37 seconds, RONI took an average of 14.80 seconds, and Chen took an average of 0.66 seconds. There is some variance depending on the dataset and the number of iterations to convergence, but \trim\ is consistently faster than other methods.}

\vspace{0.1cm}
\noindent{\bf Summary of defense results.}
\begin{itemize}
\item We find that previous defenses (RANSAC, Huber, Chen, \revision{and RONI}) do not work very well against our poisoning attacks. As seen in Figures~\ref{fig:defenseridge}-\ref{fig:defenselasso}, previous defenses can in some cases increase the MSEs over unpoisoned models.
\item Our proposed defense, \trim, works very well and significantly improves the MSEs compared to existing defenses. The median MSE increase over all attacks, models, and datasets is only 6.1\%, and only 20\% of attacks cause more than 27.2\% MSE increase. In some cases \trim\ achieves lower MSEs than those of unpoisoned models \revision{(by 3.47\%)}.
\item All of the defenses we tested ran reasonably quickly. \trim\ was the fastest, running in an average of \revision{0.02 seconds} on the house price dataset.
\end{itemize}

\vspace{0.4cm}
\section{Related work}
\label{sec:related}

The security of machine learning has received a lot of attention in different communities (e.g.,~\cite{dalvi2004adversarial,lowd2005adversarial,Huang2011adversarial,barreno2006can,biggio14-tkde,biggio18}). Different types of attacks against learning algorithms have been designed and analyzed, including  \emph{evasion attacks} (e.g.,~\cite{biggio13-ecml,Szegedy14,Goodfellow14,Srndic14,Papernot16,Papernot17,Carlini17}), and \emph{privacy attacks} (e.g.,~\cite{fredrikson2014privacy,Membership,Fredrikson15}).  In \emph{poisoning attacks} the attacker manipulates training data to violate system \emph{availability} or \emph{integrity}, \ie, to cause a denial of service or the misclassification of specific data points, respectively~\cite{Biggio2012poisoning,Huang2011adversarial,mei15-aaai,Xiao15,biggio17-aisec}.

In the security community, practical poisoning attacks have been demonstrated in  worm signature generation~\cite{Perdisci06,newsome2006paragraph}, spam filters~\cite{Nelson08}, network traffic analysis systems for detection of DoS attacks~\cite{ANTIDOTE}, sentiment analysis on social networks~\cite{Newell14}, crowdsourcing~\cite{wang14-usenix}, and  health-care~\cite{mozaffari2014systematic}. In supervised learning settings, Newsome et al.~\cite{newsome2006paragraph} have proposed \emph{red herring attacks} that add spurious words (features) to reduce the maliciousness score of an instance. These attacks work against conjunctive and Bayes learners for worm signature generation. Perdisci et al.~\cite{Perdisci06} practically demonstrate how an attacker can inject noise in the form of suspicious flows to mislead worm signature classification.
Nelson et al.~\cite{Nelson08} present both availability and targeted poisoning attacks against the public SpamBayes spam classifier. Venkataraman et al. \cite{venkataraman2008limits} analyze the theoretical limits of poisoning attacks against signature generation algorithms by proving bounds on false positives and false negatives for certain adversarial capabilities.

In unsupervised settings, Rubinstein~\etal~\cite{ANTIDOTE} examined how an attacker can  systematically inject traffic to mislead a PCA anomaly detection system for DoS attacks. Kloft and Laskov~\cite{kloft2012security} demonstrated \emph{boiling frog attacks} on centroid anomaly detection that involve incremental contamination of systems using retraining. Theoretical online centroid anomaly detection analysis has been discussed in \cite{kloft2012security}. Ciocarlie et al.~\cite{Ciocarlie08} discuss sanitization methods against time-based anomaly detectors in which multiple micro-models are built and compared over time to identify poisoned data. The assumption in their system is that the attacker only controls data generated during a limited time window.

\revision{
In the machine learning and statistics communities, earliest treatments consider the robustness of learning to noise, including the extension of the PAC model by Kearns and Li~\cite{kearns1993learning}, as well as work on robust statistics~\cite{huber2011robust,tyler2008robust,Xu10, RobustPCA}. In adversarial settings, robust methods for dealing with arbitrary corruptions of data have been proposed in the context of linear regression~\cite{Chen13}, high-dimensional sparse regression~\cite{Chen13sparse}, logistic regression~\cite{Feng14}, and linear regression with low rank feature matrix~\cite{Liu17-AISEC}. These methods are based on assumptions on training data such as sub-Gaussian distribution, independent features, and low-rank feature space.  Biggio et al.~\cite{Biggio2012poisoning} pioneered the research of optimizing poisoning attacks for kernel-based learning algorithms such as SVM.
Similar techniques were later generalized to optimize data poisoning attacks for several other important learning algorithms,
such as feature selection for classification~\cite{Xiao15}, topic modeling~\cite{poison-lda}, autoregressive models~\cite{poison-autoregressive},  collaborative filtering~\cite{Li2016data}, and simple neural network architectures~\cite{biggio17-aisec}.
}

\vspace{0.2cm}
\section{Conclusions}
\label{sec:conclusions}

We perform the first systematic study on poisoning attacks and their countermeasures for linear regression models. \revision{We propose a new optimization framework for poisoning attacks and a fast statistical attack that requires minimal knowledge of the training process.} We also take a principled approach in designing a new robust defense algorithm that largely outperforms existing robust regression methods. We extensively evaluate our proposed attack and defense algorithms on several datasets from health care, loan assessment, and real estate domains. We demonstrate the real implications of poisoning attacks in a case study health application.
We finally believe that our work will inspire future research towards developing more secure learning algorithms against poisoning attacks. 

\vspace{-0.5pt}
\section*{Acknowledgements}
We thank Ambra Demontis for confirming the attack results on ridge regression, and Tina Eliassi-Rad, Jonathan Ullman, and Huy Le Nguyen for discussing poisoning attacks. We also thank the anonymous reviewers for all the extensive feedback received during the review process.

This work was supported in part by FORCES (Foundations Of Resilient CybEr-Physical Systems), which receives support from the National Science Foundation (NSF award numbers CNS-1238959, CNS-1238962, CNS-1239054, CNS-1239166), DARPA under grant no. FA8750-17-2-0091, Berkeley Deep Drive, and Center for Long-Term Cybersecurity.

This work was also partly supported by the EU H2020 project ALOHA, under the European Union's Horizon 2020 research and innovation programme (grant no. 780788).


\appendices

\section{Theoretical Analysis of Linear Regression}
\label{app:att-th}

\revision{We prove the equivalence of $\attackloss_{\rm tr}$ and $\attackloss^\prime_{\rm tr}$ with the following theorem.}

\begin{restatable}{thm}{objequiv}
\revision{Consider OLS regression. Let $\trainset=\lbrace \Xb,Y\rbrace$ be the original dataset, $\thetapar_0 = (\wb_0,b_0)$ the parameters of the original OLS model, and $\trainset'=\lbrace \Xb,Y'\rbrace$ the dataset where $Y'$ consists of predicted values from $\thetapar_0$ on $\Xb$. Let $\poisonset=\lbrace \Xb_p,Y_p\rbrace$ be a set of poisoning points. Then
\[
\argmin_{\thetapar} \defenderloss(\trainset \cup \poisonset, \thetapar) = \argmin_{\thetapar} \defenderloss(\trainset' \cup \poisonset, \thetapar)
\]
Furthermore, we have $\frac{\partial \attackloss_{\rm tr}}{\partial \attackpar}=\frac{\partial \attackloss^\prime_{\rm tr}}{\partial \attackpar}$, where $\attackpar=(\xb_c,y_c)$. Then the optimization problem for the adversary, and the gradient steps the adversary takes, are the same whether $\attackloss_{\rm tr}$ or $\attackloss^\prime_{\rm tr}$ is used.}
\label{thm:objequiv}
\end{restatable}

\begin{proof}
\revision{We begin by showing that
\[
\argmin_{\thetapar} \defenderloss(\trainset, \thetapar) = \argmin_{\thetapar} \defenderloss(\trainset', \thetapar).
\]
By definition, we have $\thetapar_0 = \argmin_{\thetapar}\defenderloss(\trainset, \thetapar)$. In $Y'$, $y_i' = f(\xb_i,\thetapar_0)$, so $\defenderloss(\trainset', \thetapar_0)=0$. But $\defenderloss \ge 0$, so $\thetapar_0 = \argmin_{\thetapar}\defenderloss(\trainset', \thetapar)$.\\}

\revision{We can use this to show that $\Xb^T Y = \Xb^T Y'$. Recall that the closed form expression for OLS regression trained on $\Xb,Y$ is $\thetapar = (\Xb^T \Xb)^{-1} \Xb^T Y$. Because $\thetapar_0$ is the OLS model for both $\trainset,\trainset'$ , we have
\[
(\Xb^T \Xb)^{-1}(\Xb^T Y) = (\Xb^T \Xb)^{-1}(\Xb^T Y'),
\]
but $(\Xb^T \Xb)^{-1}$ is invertible, so $\Xb^T Y = \Xb^T Y'$. }
We can use this to show that $\argmin_{\thetapar} \defenderloss(\trainset \cup \poisonset, \thetapar)= \argmin_{\thetapar} \defenderloss(\trainset' \cup \poisonset, \thetapar)$ for any $\poisonset$. Consider the closed form expression for the model learned on $\trainset \cup \poisonset$:

\begin{eqnarray*}
& (\Xb^T \Xb + \Xb^T_p X_p)^{-1}(\Xb^T Y + \Xb^T_p Y_p) \\
= &  (\Xb^T \Xb + \Xb^T_p \Xb_p)^{-1} \cdot (\Xb^T Y' + \Xb^T_p Y_p) \\
\end{eqnarray*}
\noindent which is exactly the model learned on $\trainset' \cup \poisonset$. So the learned models for the two poisoned datasets are the same. Note that this also holds for ridge regression, where the Hessian has a $\lambda I$ term added, so it is also invertible.

\revision{We proceed to use $\Xb^T Y = \Xb^T Y'$ again to show that $\frac{\partial \attackloss_{\rm tr}}{\partial \attackpar}=\frac{\partial \attackloss^\prime_{\rm tr}}{\partial \attackpar}$.
\[
\frac{\partial \attackloss_{\rm tr}}{\partial \attackpar} = \frac{2}{n}(\Xb \thetapar-Y)^T\Xb\frac{\partial\thetapar}{\partial \attackpar}
\]
\[
\frac{\partial \attackloss^\prime_{\rm tr}}{\partial \attackpar} = \frac{2}{n}(\Xb \thetapar-Y')^T\Xb\frac{\partial\thetapar}{\partial \attackpar}
\]
So the difference between the gradients is
\[
\frac{\partial \attackloss^\prime_{\rm tr}}{\partial \attackpar}-\frac{\partial \attackloss_{\rm tr}}{\partial \attackpar}=\frac{2}{n}(Y-Y')^T\Xb\frac{\partial\thetapar}{\partial \attackpar}=\mathbf{0}.
\]
Then both the learned parameters and the gradients of the objectives are the same regardless of the poisoned data added.
}
\end{proof}

\revision{We can now perform the derivation of the exact form of the gradient of $\attackloss^\prime_{\rm tr}$. We have:
\[
\frac{\partial \attackloss^\prime_{\rm tr}}{\partial \attackpar} = \frac{2}{n}\sum_{i=1}^n((\wb-\wb_0)^T\xb_i+(b-b_0))\left(\xb_i^T\frac{\partial\wb}{\partial \attackpar} + \frac{\partial b}{\partial \attackpar}\right).
\]
The right hand side can be rearranged to
\[
(\wb-\wb_0)^T\left(\mathbf{\Sigma}\frac{\partial\wb}{\partial \attackpar} +\mathbf{\mu}\frac{\partial b}{\partial \attackpar}\right) + (b-b_0)\left(\mathbf{\mu}^T\frac{\partial\wb}{\partial \attackpar} +\frac{\partial b}{\partial \attackpar}\right),
\]
but the terms with gradients can be evaluated using the matrix equations derived from the KKT conditions from Equation~\ref{eq:grad-theta-yc}, which allows us to derive the following:
\begin{eqnarray*}
\frac{\partial \attackloss^\prime_{\rm tr}}{\partial\xb_c} &=& \frac{2}{n}((\wb_0-\wb)^T M+(b_0-b)\wb^T\\
 &=& \frac{2}{n}(f(\xb_c,\thetapar)-f(\xb_c,\thetapar_0))(\wb_0-2\wb)^T\\
\frac{\partial \attackloss^\prime_{\rm tr}}{\partial y_c} &=& \frac{2}{n}(f(\xb_c,\thetapar) - f(\xb_c,\thetapar_0)).
\label{eqn:linreg-grad}
\end{eqnarray*}}

\section{Analysis of \trim\ algorithm}
\label{app:analysis}

We present here an analysis on the convergence and estimation properties of the \trim\
algorithm.

\vspace{0.1cm}
\noindent {\bf Convergence.} First, Algorithm~\ref{alg:trim} can be proven to always terminate by the following theorem.

\trimconv* 

\begin{proof} We first prove that for each iteration $i$ that does not terminate the loop, we have $R^{(i)}<R^{(i-1)}$. Since each subset of $\{1,...,n\}$ with size $n-\npois$ uniquely corresponds to one value $R$, there is only a finite number of possible $R$ during training. If the algorithm does not terminate, then there will be an infinite long sequence of $R^{(i)}$, contradicting that the set of all possible $R$ is finite.

We only need to show $R^{(i)}\leq R^{(i-1)}$, as the algorithm terminates when $R^{(i)}= R^{(i-1)}$. In fact, we have
\begin{eqnarray}
R^{(i)}=\defenderloss(\alltrain^{\mathcal{I}^{(i)}}, \thetapar^{(i-1)}) & \leq & \defenderloss(\alltrain^{\mathcal{I}^{(i-1)}}, \thetapar^{(i-1)})\nonumber\\
&\leq&  \defenderloss(\alltrain^{\mathcal{I}^{(i-1)}}, \thetapar^{(i-2)})=R^{(i-1)}.\nonumber
\end{eqnarray}
The first inequality is because of the definition of $\mathcal{I}^{(i)}$ (line 8), while the second is due to the definition of $\thetapar^{(i)}$ (line 9).
\end{proof}

\vspace{0.1cm}
\noindent {\bf Estimation bound.} We now prove Theorem~\ref{thm:2}. We restate it below.

\addtocounter{thm}{-2}

\begin{thm} Let \trainset\ denote the original training data, $\hat{\thetapar}$ the global optimum for (\ref{eq:opt}), and $\thetaopt=\argmin_{\thetapar} \defenderloss(\trainset, \thetapar)$ the estimator in the ideal world on pristine data. Assuming $\alpha<1$, there exist a subset $\mathcal{D}' \subset \trainset$ of $(1-\alpha)\cdot n$ pristine data samples such that
\begin{equation}
\mathrm{MSE}(\mathcal{D}',\hat{\thetapar}) \leq \bigg( 1+\frac{\alpha}{1-\alpha}\bigg)  \defenderloss(\trainset,\thetaopt)
\end{equation}
\end{thm}

\begin{proof} Assume $\hat{\thetapar}=(\hat{\wb}, \hat{b}), \hat{\mathcal{I}}$ optimize (\ref{eq:opt}). We have:
\begin{equation}
  \defenderloss(\alltrain^{\hat{\mathcal{I}}}, \hat{\thetapar}) \leq \defenderloss(\trainset, \thetaopt).
  \label{eq:1}
\end{equation}

Since the adversary can poison at most $\alpha\cdot n$ data points, there exists a subset $\mathcal{I}'\subseteq\hat{\mathcal{I}}$ containing $(1-\alpha)n$ indexes corresponding to pristine data points. We define $\mathcal{D}'=\alltrain^{\mathcal{I}'}$.
Thus, we have
\begin{eqnarray}
\defenderloss(\mathcal{D}', \hat{\thetapar}) &=& \frac{1}{(1-\alpha)n} \sum_{(\xb, y)\in\mathcal{D}'} (\hat{\wb}^T \xb +\hat{b} - y)^2 + \lambda\Omega(\hat{\thetapar})\nonumber\\
&\leq& \frac{1}{(1-\alpha)n} \sum_{(\xb, y)\in\alltrain^\mathcal{I}} (\hat{\wb}^T \xb +\hat{b} - y)^2 + \lambda\Omega(\hat{\thetapar})\nonumber\\
&=& \frac{1}{1-\alpha} \defenderloss(\alltrain^{\hat{\mathcal{I}}}, \hat{\thetapar}) - \frac{1}{1-\alpha}\cdot\lambda\Omega(\hat{\theta}) +\lambda\Omega(\hat{\theta})\nonumber\\
&\leq& \frac{1}{1-\alpha} \defenderloss(\trainset, \thetaopt) - \frac{\alpha}{1-\alpha}\cdot\lambda\Omega(\hat{\thetapar})\nonumber\\
&\leq&  \bigg( 1 + \frac{\alpha}{1-\alpha}\bigg) \defenderloss(\trainset, \thetaopt). \label{eq:2}
\end{eqnarray}
Notice that in the second step, we apply the fact below:
\[\sum_{(\xb, y)\in\alltrain^\mathcal{I}} (\hat{\wb}^T \xb + \hat{b} - y)^2 = n [\defenderloss(\alltrain^\mathcal{I},\hat{\thetapar}) - \lambda\Omega(\hat{\thetapar})]\]
The second to last step is derived by applying Inequality (\ref{eq:1}), and the last step comes from
\[\lambda\Omega(\hat{\theta})\geq 0.\]
Further, we have
\begin{equation}
\mathrm{MSE}(\mathcal{D}', \hat{\thetapar})\leq \defenderloss(\mathcal{D}', \hat{\thetapar})\label{eq:4}
\end{equation}
By combining (\ref{eq:2}) and (\ref{eq:4}), we can get our conclusion.
\end{proof}

\section{Baseline attack}
\label{sec:optattacks}

\begin{table}
\begin{center}
\begin{tabular}{|c|c|c|}
\hline
{\small Parameter} & Purpose & Values \\
\hline
$\eta$ & Line Search Learning Rate & $\lbrace 0.01,0.03,0.05,0.1,$  \\
--- & --- & $0.3,0.5,1.0\rbrace$\\
$\beta$ & Line Search Learning Rate Decay & $\lbrace 0.75\rbrace$ \\
$\epsilon$ & Attack Stopping Condition & $10^{-5}$ \\
$\lambda$ & Regularization Parameter & {\scriptsize Set with Cross Validation} \\
\hline
\end{tabular}
\end{center}
\caption{Description of Parameters for Algorithm ~\ref{alg:poisoning}.}
\label{tab:pdfparams}
\end{table}
In this section, we discuss parameter setting for the baseline attack by Xiao et al.~\cite{Xiao15}. We perform experiments on the same dataset used by Xiao et al.~\cite{Xiao15} to test and optimize the baseline attack.

\vspace{0.02cm}
\noindent {\bf PDF dataset.} The PDF malware dataset is a classification dataset containing 5000 benign and 5000 malicious PDF files~\cite{Srndic14,contagio}. It includes 137 features, describing information such as size, author, keyword counts, and various timestamps. We pre-process the data by removing features that were irrelevant (e.g., file name) or had erroneous values (e.g., timestamps with negative values). We also use one-hot encoding for categorical features (replacing the categorical feature with a new binary feature for each distinct value) and apply the logarithm function to columns with large values (e.g., size, total pixels), resulting in 141 features.  Each feature is normalized by subtracting the minimum value, and dividing by its range, so that all these features are in [0,1].

\vspace{0.02cm}
\noindent {\bf Hyperparameters.} In order to analyze the baseline attack, we perform an experiment that reproduces exactly the setting from Xiao et al.~\cite{Xiao15}. We choose a random subset of 300 files for training and a non-overlapping subset of 5000 points for testing the models. To take advantage of our multi-core machines, we parallelize the code by allowing each core to run different instances of the {\bf for} loop body starting on line 6 in Algorithm~\ref{alg:poisoning}.

There are 3 hyperparameters that control the gradient step and convergence of the iterative search procedure in the algorithm ($\eta, \beta$,  and $\epsilon$). The $\eta$ parameter controls the step size taken in the direction of the gradient. We selected from 7 different values in a wide range, by testing each on 20\% poisoning and identifying the value with the largest MSE increase. The $\beta$ parameter controls the decay of the learning rate, so we take smaller steps as we get closer to the optimal value. We fixed this value to 0.75 and decayed (set $\eta\leftarrow \eta*\beta$) when a step did not make progress. We found this setting to work well on many problems. We fixed the $\epsilon$ parameter for attack stopping condition at 0.00001, and choose the $\lambda$ regularization parameter for the regression model with cross validation. Our parameter settings are detailed in Table~\ref{tab:pdfparams}).

\begin{figure}[tbhp]
\centering
\includegraphics[width=2.5in]{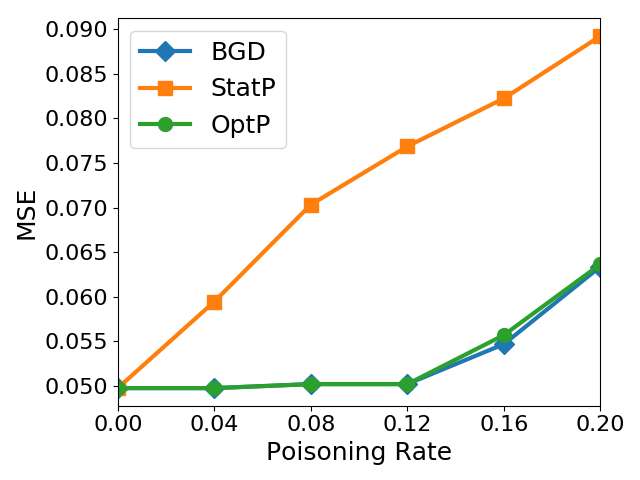}
\caption{\revision{Attack MSE on Contagio dataset for ridge.}}
\label{fig:contagio-ridge}
\end{figure}

\begin{figure}[h]
\centering
\includegraphics[width=2.5in]{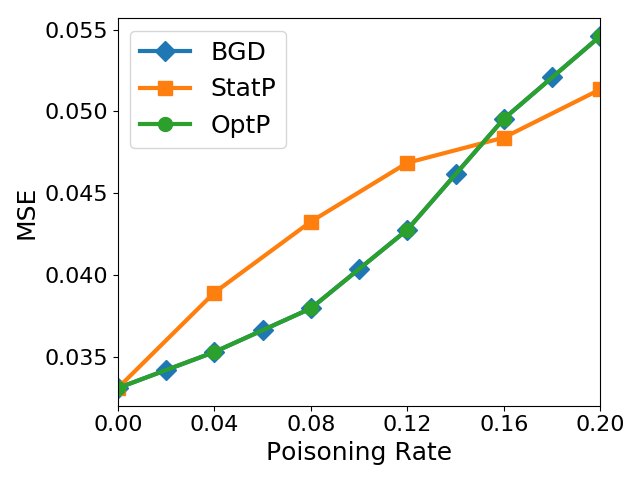}
\caption{\revision{Attack MSE on Contagio dataset for Lasso.}}
\label{fig:contagio-lasso}
\end{figure}

\vspace{0.02cm}
\noindent {\bf Attack effectiveness.}  \revision{We ran our \optp\ and \statp\ attacks on this dataset, in addition to \baseline. We expect \optp\ to be very similar in terms of poisoning to \baseline\ because it is run in the classification setting. Our proposed optimization framework is specific to regression. For instance, in classification settings optimizing by both $x$ and $y$ variables is exactly the same as optimizing only by $x$. For the initialization strategies, \invflip\ and \bflip\ are  exactly the same in the classification setting. In our framework we are exploiting the continuous response variables of regression for designing more effective optimization-based poisoning attacks. The only modification to \baseline\ might come from using \valobj\ as an optimization objective, but we expect that in isolation that will not produce significant changes. We showed in Section~\ref{sec:exp-new-attacks} that \valobj\ is most likely to be effective when optimization by $(x,y)$ is used. Our graphs from Figures~\ref{fig:contagio-ridge} and \ref{fig:contagio-lasso} confirm this expectation, and indeed \optp\ and \baseline\ are very similar in the attack MSEs. The effectiveness of the \baseline\ attack is similar to that reported by Xiao et al.~\cite{Xiao15} and we have confidence that our implementation and choice of hyper-parameters are accurate. Interestingly, the \statp\ attack outperforms \baseline\ and \optp\ by 40\% for ridge regression. We believe that pushing the feature values to the boundary as done by \statp\ has higher effect as a poisoning strategy for ridge regression in which the loss function is convex and the optimization maximum is achieved in the corners. That is not always the case with models such as Lasso, but still \statp\ is quite effective at poisoning Lasso as well.}

\end{document}